\documentclass[letterpaper, 10 pt, conference]{ieeeconf}  

\IEEEoverridecommandlockouts                              

\usepackage{graphicx}
\usepackage{epsfig} 
\usepackage{mathptmx}
\usepackage{amsmath} 

\usepackage{amssymb}  
\usepackage{amsthm}
\usepackage{amsfonts}
\usepackage{wasysym}
\usepackage[mathcal]{euscript}
\usepackage{bbm}
\usepackage{lipsum}
\usepackage[ruled,vlined,linesnumbered]{algorithm2e}
\usepackage[colorlinks=true,linkcolor=black,anchorcolor=black,citecolor=black,filecolor=black,menucolor=black,runcolor=black,urlcolor=black]{hyperref}
\usepackage{etoolbox}
\usepackage{multirow}
\usepackage{mathtools}
\usepackage{siunitx}
\usepackage[labelformat=simple]{subcaption}

\usepackage{outlines}
\let\labelindent\relax
\usepackage{enumitem}
\setenumerate[2]{label=\alph*.}
\setenumerate[3]{label=\roman*.}

\usepackage{pifont}

\usepackage{booktabs}
\usepackage{textcomp}

\usepackage{marginnote}

\makeatletter
\patchcmd{\@makecaption}
  {\scshape}
  {}
  {}
  {}
\makeatother

\newtheorem{defn}{Definition}

\newtheorem{lem}[defn]{Lemma}
\newtheorem{prop}[defn]{Proposition}

\newtheorem{cor}[defn]{Corollary}

\providecommand{\R}{\ensuremath \mathbb{R}}
\providecommand{\N}{\ensuremath \mathbb{N}}
\newcommand{\idxset}{\mathcal{I}}

\newcommand{\pN}{\mathcal{N}}

\newcommand{\regtext}[1]{\mathrm{\textnormal{#1}}}

\newcommand{\defemph}[1]{\emph{#1}}
\newcommand{\ts}[1]{\textsuperscript{#1}}

\DeclareMathOperator*{\argmax}{arg\,max}

\newcommand{\card}[1]{\left\vert#1\right\vert}

\newcommand{\norm}[1]{\left\Vert#1\right\Vert}
\newcommand{\abs}[1]{\left\vert#1\right\vert}
\newcommand{\pow}[1]{\regtext{pow}\left(#1\right)}
\newcommand{\diag}[1]{\regtext{diag}\!\left(#1\right)}

\newcommand{\eig}[1]{\regtext{eig}\!\left(#1\right)}
\newcommand{\union}{\bigcup}

\newcommand{\trans}{^\top}
\newcommand{\inv}{^{-1}}
\newcommand{\pinv}{^{\dagger}}

\newcommand{\bigO}{\mathcal{O}}
\newcommand{\bdry}[1]{\partial#1}
\renewcommand{\ker}[1]{\regtext{ker}\!\left(#1\right)}
\newcommand{\convhull}[1]{\regtext{CH}\!\left(#1\right)}

\newcommand{\emptyarr}{[\ ]}
\newcommand{\zeros}{{0}}
\newcommand{\ones}{{1}}
\newcommand{\eye}{\regtext{I}}
\newcommand{\coef}{\beta}
\newcommand{\slackvar}{\coef_{\regtext{s}}}

\newcommand{\zonofn}[1]{\mathcal{Z}\!\left(#1\right)}
\newcommand{\ellipsoidfn}[1]{\mathcal{E}\!\left(#1\right)}
\newcommand{\conzonofn}[1]{\mathcal{CZ}\!\left(#1\right)}

\newcommand{\conpolyzonofn}[1]{\mathcal{CPZ}\!\left(#1\right)}
\newcommand{\ellifn}[2]{\mathcal{E}_{#1}\left(#2\right)}
\newcommand{\ballfn}[2]{\mathcal{B}_{#1,#2}}
\newcommand{\ballprodfn}[1]{\mathcal{B}_\times\!\left(#1\right)}
\newcommand{\hpfn}[1]{\mathcal{P}\!\left(#1\right)}
\newcommand{\hspacefn}[1]{\mathcal{H}\!\left(#1\right)}

\newcommand{\tobdryfn}[2]{\regtext{bdproj}_{#1}\!\left(#2\right)}
\newcommand{\costfunc}[1]{\regtext{cost}\!\left(#1\right)}
\newcommand{\rayfn}[1]{\mathcal{R}\!\left(#1\right)}

\newcommand{\ch}{_\regtext{\tiny{CH}}}
\newcommand{\rob}{_{\regtext{rob}}}
\newcommand{\uncrt}{_{\regtext{unc}}}

\newcommand{\reduce}{_{\regtext{r}}}
\newcommand{\keep}{_{\regtext{keep}}}

\newcommand{\ndim}{n}
\newcommand{\ngen}{m}
\newcommand{\ncon}{k}
\newcommand{\nsum}{v}

\newcommand{\nplot}{{n_{\regtext{plot}}}}

\newcommand{\idx}[1]{{\langle#1\rangle}} 

\newcommand{\st}{\regtext{ s.t. }}

\usepackage[
    style=ieee,
    doi=false,
    isbn=false,
    url=false,
    eprint=false,
    backend=biber,
    natbib=true,
    maxbibnames=10
    ]{biblatex}
    
\bibliography{references}

\newcommand{\new}[1]{#1}
\newcommand{\newnc}[1]{#1}

\begin{document}

\title{Ellipsotopes: Uniting Ellipsoids and Zonotopes for Reachability Analysis and Fault Detection
}

\author{Shreyas Kousik$^1$ ,
Adam Dai$^2$,
and Grace X. Gao$^1$ 
\thanks{$^1$ Aeronautics and Astronautics, Stanford University, Stanford, CA.}
\thanks{$^2$ Electrical Engineering, Stanford University, Stanford, CA.}
\thanks{Corresponding author: \texttt{gracegao@stanford.edu}.}
}

\maketitle

\thispagestyle{plain}
\pagestyle{plain} 

\begin{abstract}
Ellipsoids are a common representation for reachability analysis, because they can be transformed efficiently under affine maps, and allow conservative approximation of Minkowski sums, which let one incorporate uncertainty and linearization error in a dynamical system by expanding the size of the reachable set.
Zonotopes, a type of symmetric, convex polytope, are similarly frequently used due to efficient numerical implementation of affine maps and exact Minkowski sums.
Both of these representations also enable efficient, convex collision detection for fault detection or formal verification tasks, wherein one checks if the reachable set of a system collides (i.e., intersects) with an unsafe set.
However, both representations often result in conservative representations for reachable sets of arbitrary systems, and neither is closed under intersection.
Recently, representations such as constrained zonotopes and constrained polynomial zonotopes have been shown to overcome some of these conservativeness challenges, and are closed under intersection.
However, constrained zonotopes can not represent shapes with smooth boundaries such as ellipsoids, and constrained polynomial zonotopes can require solving a non-convex program for collision checking or fault detection.
This paper introduces \textit{ellipsotopes}, a set representation that is closed under affine maps, Minkowski sums, and intersections.
Ellipsotopes combine the advantages of ellipsoids and zonotopes while ensuring convex collision checking.
The utility of this representation is demonstrated on several examples.
\end{abstract}

\renewcommand{\qedsymbol}{$\blacksquare$}

\section{Introduction}

In the controls, robotics, and navigation communities, it is often critical to place strict guarantees on the behavior of a dynamical system.
Example applications of such guarantees include collision avoidance \cite{kousik2020bridging_ijrr,althoff2010_dissertation,shetty2020_stoch_reach,chen2021fastrack}, fault detection \cite{bhamidipati2020_stoch_reach,scott2016constrained_zonotopes}, and control invariance \cite{ames2016control,smit2019walking,chen2021fastrack}.
A common strategy for enforcing guarantees is to compute the system's reachable set of states, then check that system measurements lie within this set (e.g., for fault detection) or the set obeys non-intersection constraints (e.g., for collision avoidance).
Directly representing a continuum of possible system trajectories numerically is typically intractable, given that these trajectories are solutions to a nonlinear differential or difference equation.
Instead, a variety of set representations have been introduced to enable approximating reachable sets.
Two of the most common and well-studied representations are ellipsoids \cite{kurzhanski2000ellipsoidal,kurzhanskiy2006ellipsoidal_toolbox} and zonotopes \cite{combastel2005zono_state_observer,girard2005reachability_zono,althoff2010_dissertation}.
In this work, an ellipsoid is best understood as an affine transformation of a unit 2-norm ball in an arbitrary-dimensional Euclidean space.
A zonotope can similarly be understood as the affine transformation of the unit $\infty$-norm ball, resulting in a symmetric polytope.
We propose a set representation, \defemph{ellipsotopes}, by generalizing to arbitrary $p$-norms, as shown in Fig. \ref{fig:ellipsotope_increasing_norm}.

\begin{figure}[t]
    \centering
    \includegraphics[width=0.8\columnwidth]{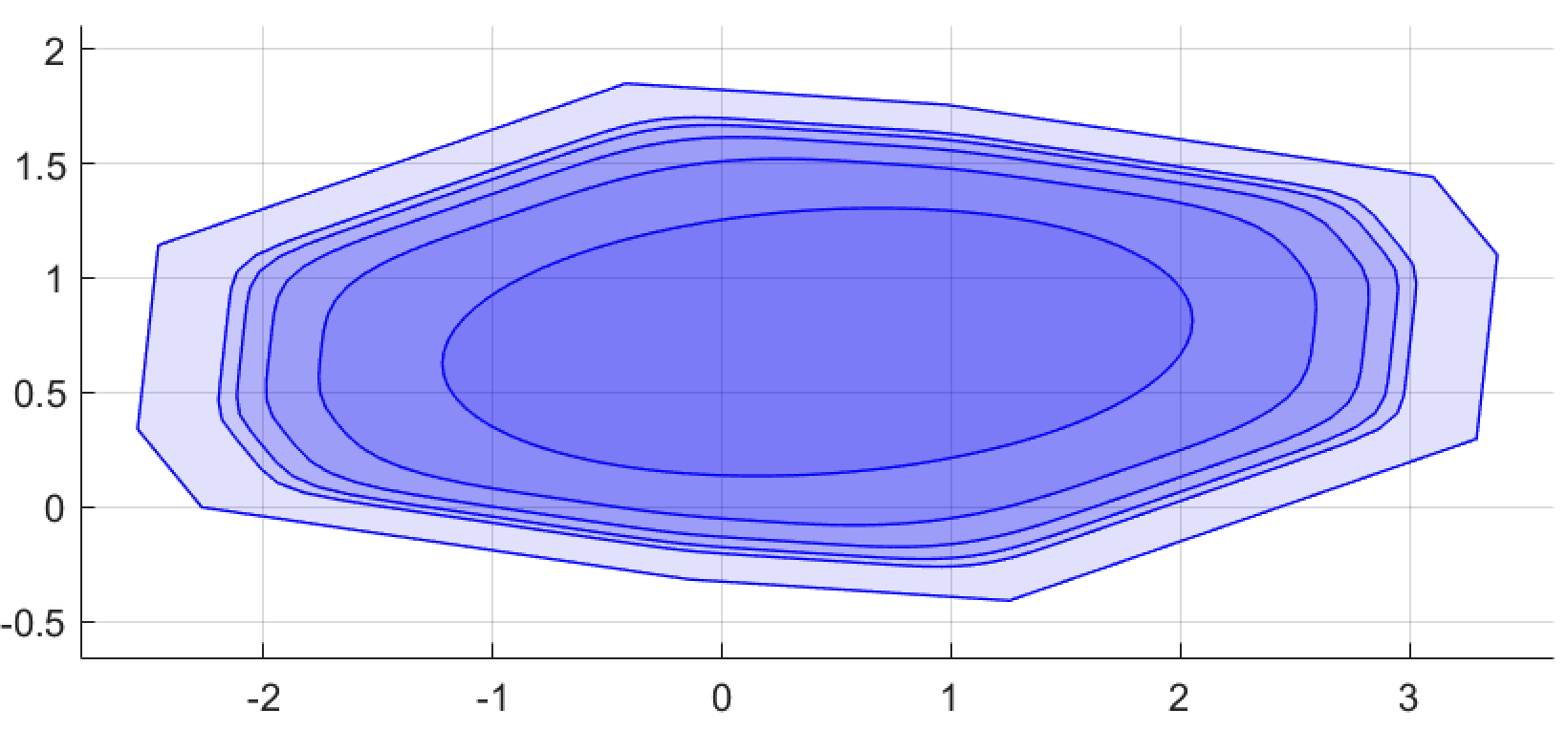}
    \caption{Basic ellipsotopes with five generators and increasing $p$-norm ($p = 2,4,\cdots,10$), shown with lighter blue as the norm increases.
    The outermost shape is the $\infty$-norm zonotope, and the innermost shape is the 2-norm ellipsoid.}
    \label{fig:ellipsotope_increasing_norm}
    \vspace*{-0.3cm}
\end{figure}

\subsection{Ellipsoids and Zonotopes}

Both ellipsoids and zonotopes provide straightforward numerical implementations of operations that are commonly-used for reachability analysis, fault detection, and similar tasks.
For example, both representations can be transformed readily via affine maps, thereby representing the flow of a (linearized) dynamical system.
Furthermore, one can apply convex programming to efficiently detect when these sets intersect with e.g., obstacles for collision avoidance \cite{guibas2003zonotopes,scott2016constrained_zonotopes,kurzhanskiy2006ellipsoidal_toolbox}.
However, choosing between the two representations comes with certain tradeoffs.
For example, zonotopes are closed under Minkowski sums, which are used to incorporate uncertainty and linearization error, while ellipsoids are not.
On the other hand, ellipsoids can exactly represent confidence level sets of Gaussian distributions, while zonotopes cannot. 

Note, we present a more detailed discussion of other set representations, both convex and non-convex, in Section \ref{sec:related_work}.
Out of the convex representations, we consider ellipsoids and zonotopes the best-suited for reachability and fault detection tasks.
For the non-convex representations, we typically lose the ability to perform efficient, convex collision-checking.

\subsection{Contributions and Paper Organization}

Our main contribution is a novel set representation called the \defemph{ellipsotope}, which combines the advantages of both ellipsoids and zonotopes at the expense of slightly more computational cost, but lower conservativeness.
This is useful, for example, when some uncertainties in a dynamical system are Gaussian (i.e., ellipsoidal) whereas other uncertainties are bounded and polytopic, as often arises in autonomous navigation \cite{shetty2020_stoch_reach,bhamidipati2020_stoch_reach,althoff2010_dissertation}.
For the purposes of reachability analysis and fault detection, we show that ellipsotopes are closed under linear maps, Minkowski sums, and intersections.
We also present order reduction strategies for managing ellipsotope complexity, which can grow during reachability analysis.
We demonstrate ellipsotopes' utility, and their advantages over ellipsoids or zonotopes alone, with several numerical examples.
The code for our examples is open source\footnote{\url{https://github.com/Stanford-NavLab/ellipsotopes}}.


Section \ref{sec:related_work} discusses a variety of set representations to clarify the context for ellipsotopes.
Section \ref{sec:preliminaries} introduces notation and set representations relevant to developing ellipsotopes.
Section \ref{sec:etopes} defines ellipsotopes and discusses properties and closed operations.
Section \ref{sec:order_reduction} presents methods for reducing ellipsotope complexity, which may grow due to the operations in Sec. \ref{sec:etopes}.
Section \ref{sec:numerical_examples} covers numerical examples and applications of ellipsotopes.
Section \ref{sec:conclusion} concludes the paper.
We provide additional properties of ellipsotopes along with strategies for visualization in the appendix.
\section{Related Work}\label{sec:related_work}

A variety of convex and non-convex set representations exist for reachability analysis and fault detection.
We now briefly discuss these representations, and under which operations they are \defemph{closed}, meaning that an operation yields an instance of the same representation.
\new{Note, a thorough review is available \cite{althoff2021set}.}

\subsection{Convex Set Representations}

Convex representations enable one to use convex programming with convergence guarantees to evaluate intersection and set membership, e.g., to check if a system's reachable set lies within a safe region.
In particular, we discuss ellipsoids, convex polytopes, and support functions.

As mentioned before, ellipsoids are affine transformations of the unit 2-norm ball.
This set representation is closed under affine transformations and hyperplane intersections \cite{kurzhanski2000ellipsoidal}.
For operations such as Minkowski sum, intersection, Pontryagin (Minkowski) difference, and convex hull, efficient algorithms exist to generate inner- and outerapproximative ellipsoids \cite{kurzhanskiy2006ellipsoidal_toolbox,halder2018parameterized_ellipsoid_approx,yildirim2006minimum_vol_ellipsoid_convhull}.
Most importantly, for tasks such as reachability analysis, confidence level sets of multivariate Gaussian distributions are ellipsoidal.
Unfortunately, ellipsoidal representations of reachable sets can rapidly become conservative due to the overapproximation required for Minkowski sums.
Furthermore, ellipsoids are not well-suited to representing polytopic sets such as occupancy grids, which are commonly used for tasks such as robot motion planning.

Convex polytopes can be thought of as the bounded intersection of a collection of affine halfspaces in arbitrary dimensions (H-representation); note, an unbounded intersection is called a polyhedron \cite{sadraddini2019linear_polytope_containment}.
Another common representation is as the convex hull of a set of vertices (V-representation).
This broad category of objects is closed under Minkowski sum, intersection, Pontryagin difference, and convex hull \cite{kvasnica2004mpt_toolbox}.
The H-representation is especially convenient for determining if a polytope contains a point and performing intersections.
However, the remaining operations are not computationally efficient, especially  in high dimensions or when a convex polytope is defined by a large number of halfspaces.

To avoid these challenges, zonotopes have become a popular representation that enable efficient Minkowski sums and set containment queries \cite{combastel2005zono_state_observer,girard2005reachability_zono,althoff2010_dissertation,guibas2003zonotopes,sadraddini2019linear_polytope_containment}.
A zonotope is a centrally-symmetric convex polytope constructed as a Minkowski sum of line segments.
Zonotopes can be parameterized by a center and generator (see \eqref{eq:zono_defn} in Section \ref{sec:preliminaries}), which is called a CG-representation; any point in the zonotope is the center plus a linear combination of the generators, each scaled by a coefficient in $[-1,1]$.
\new{For zonotopes, set containment (checking if one set is a subset of another) can be approximated in general \cite{sadraddini2019linear_polytope_containment} or solved in polynomial time by fixing the number of generators of one input zonotope \cite{kulmburg2021co}.
Since zonotopes are not closed under intersection or Pontryagin difference, researchers have introduced zonotope bundles \cite{althoff2011zonotope}, AH-polytopes \cite{sadraddini2019linear_polytope_containment} and constrained zonotopes \cite{scott2016constrained_zonotopes}.
A zonotope bundle stores each zonotope participating in an intersection.
An AH-polytope is the affine transformation of an H-representation of a polytope (e.g., a zonotope is the affine transformation of a hypercube).
A constrained zonotope is a zonotope with additional linear constraints on its coefficients, and can represent any convex polytope \cite[Thm. 1]{scott2016constrained_zonotopes}}.
These representations are closed under affine transformation, Minkowski sum, intersection, and, for constrained zonotopes, Pontryagin difference and convex hull \cite{raghuraman2020set_ops_conzono}.
Set membership or intersection can be evaluated with linear programming \cite{scott2016constrained_zonotopes}.
While zonotope bundles, AH-polytopes, and constrained zonotopes overcome many of the challenges of zonotopes, they cannot represent sets with curved boundaries.

Support functions enable one to represent arbitrary convex sets, allowing generalization beyond polytopes and ellipsoids \cite{girard2008efficient_reach_spt_func,le2009reachability_support_funcs,le2010reachability_support_funcs}.
A support function is a convex function that maps a vector in Euclidean space to the maximum dot product between that vector and any element in a convex set, thus representing the set implicitly.
Support functions of many convex sets, such as unit balls, ellipsoids, and zonotopes, have a simple analytical form, and support functions of polytopes can be expressed as the solution of a linear program \cite{le2009reachability_support_funcs}.
Furthermore, affine maps, Minkowski sums, and convex hulls have analytic formulations.
Unfortunately, the intersection of sets represented by support functions can only be overapproximated and may be non-convex \cite[Prop. 4]{le2010reachability_support_funcs}, so using intersection for collision-checking and fault detection is neither straightforward nor conservative.


\subsection{Non-Convex Set Representations}

The reachable set of a dynamical system is not necessarily convex.
Furthermore, robots and other autonomous systems frequently have non-convex bodies, and such systems are not necessarily subject to convex constraints for fault detection or collision avoidance.
A variety of non-convex set representations exist that attempt to address these challenges.
In particular, we discuss polynomial zonotopes, star sets, level sets, and constructive solid geometry (CSG).

Polynomial zonotopes (PZs) are a generalization of zonotopes wherein the coefficients of a zonotope's generators are instead monomials \cite{althoff2013poly_zono,kochdumper2020sparse_poly_zono}.
\new{By leveraging a center/generator structure, these sets are closed under affine transformation, Minkowski sum, convex hull.
One can add polynomial constraints on the coefficients to make constrained polynomial zonotopes (CPZs), which are additionally closed under intersections and unions \cite{kochdumper2020cons_poly_zono}.}
PZs and CPZs provide less conservative approximations of reachable sets than zonotopes, at the expense of being non-convex (so, collision checking requires solving a non-convex program).
One alternative is to overapproximate a PZ or CPZ with a zonotope \cite{holmes2020reachable}, resulting in a convex collision check at the expense of conservativeness.

Star sets also generalize zonotopes and ellipsoids to instead use a logical predicate constraint on the generator coefficients \cite{duggirala2016parsimonious_star_set,bak2017simulation_star_set,tran2019star_sets}.
\new{These sets can be non-convex, and are closed under affine transformation, Minkowski sum, and intersection; but, intersections may not be algorithmically tractable for arbitrary logical predicates \cite{althoff2021set}.
Similarly, collision checking may require solving a non-convex problem.}

Departing from center/generator representations, level sets are a popular representation for reachability analysis, because arbitrary sets can be represented as the 0-sublevel set of a function.
Such a function can be approximated on a grid \cite{mitchell2005hjb,mitchell2007ls_toolbox} or as a polynomial \cite{lasserre2009moments}.
Level sets can be used to conservatively compute reachable sets of dynamic systems subject to uncertainty \cite{majumdar2014convex,chen2021fastrack,kousik2020bridging_ijrr,holmes2016convex}.
In special cases, one can represent collision checking as a polynomial evaluation \cite{kousik2020bridging_ijrr}; in general, Minkowski sums, intersections, and convex hulls can be approximated using sums-of-squares programming.
Level set methods typically do not require linear maps and Minkowski sums for reachability analysis, instead requiring one to approximately solve a partial differential equation.
Furthermore, they suffer the curse of dimensionality for nonlinear systems with more than 5 dimensions \cite{chen2021fastrack,kousik2020bridging_ijrr}.

Constructive solid geometry (CSG) is used to model non-convex shapes in computer graphics by leveraging implicit point membership classification functions to express geometric primitives such as spheres, prisms, and cones \cite{requicha1977constructive,foley1996computer_graphics}.
Non-convex bodies are represented as unions, intersections, and set differences of primitives, which can also be approximated with smooth functions \cite{lutz2021efficient_optimization_based_collision_avoidance}.
This representation has been applied to reachability, with similar advantages and drawbacks to support functions \cite{mitchell2012csg_with_ls_toolbox,lutz2021efficient_optimization_based_collision_avoidance}.
For these sets, computing Minkowski sums is challenging; furthermore, these representations are typically limited to 2-D or 3-D settings, and it is unclear how to reduce the growing complexity of a reachable set in a similar way to zonotope order reduction.

\subsection{Summary}

From this review of a wide variety of representations, we identify several advantages and challenges.
The advantages of zonotopes and similar objects is their numerical simplicity for representing affine transformations, Minkowski sums, and collision/emptiness checking (via intersection and convex programming).
The challenges are to represent smooth or non-polytopic sets without incurring conservativeness (as with ellipsoids) or non-convexity (as with polynomial zonotopes).
Our proposed ellipsotope representation directly addresses this tradeoff by enabling efficient reachability and fault detection operations for both polytope-like and ellipsoid-like objects without introducing challenges from losing convexity.
\new{In particular, when sets are given as both polytopes and ellipsoids (see \cite{shetty2020_stoch_reach,gassmann2020scalable_ellipsoid_zono_conversion} as examples), we can represent them as ellipsotopes, then propagate and manipulate via the operations outlined in this paper either conservatively or exactly, while always ensuring convex collision checking.}

\section{Preliminaries}\label{sec:preliminaries}

We now introduce notation and several set representations.

\subsection{Notation}\label{subsec:notation}

Scalars and vectors are lowercase and italic.
Sets and matrices are uppercase italic.
The real numbers are $\R$, and the natural numbers are $\N$.
If $n \in \N$, we denote $\N_n = \{1,2,\cdots,n\} \subset \N$.
The $p$-norm unit ball in $\R^n$ is 
\begin{align}
    \ballfn{p}{n} = \{x \in \R^n\ |\ \norm{x}_p \leq 1\}.
\end{align}
An affine subspace (i.e., affine hyperplane) of $\R^n$ parameterized by $H \in \R^{k\times n}$, $k\in\N$, $k < n$, and $f \in \R^m$ is
\begin{align}
    \hpfn{H,f} = \left\{x \in \R^n\ |\ Hx = f \right\}.
\end{align}
A halfspace parameterized by $h \in \R^n$ and $s \in \R$ is 
\begin{align}
    \hspacefn{h,s} = \left\{x \in \R^n\ |\ h\trans x \leq s \right\}.    
\end{align}

Let $A$ be a set such that $A \subset \R^n$.
Its power set is $\pow{A}$, its cardinality is $\card{A}$ and its boundary is $\bdry{A}$.
Let $B \subset \R^n$ as well.
The Minkowski sum is $A \oplus B = \{a+b\ |\ a \in A,\ b \in B\}$.

Consider a set of integers $J = \{j_1,j_2,\cdots,j_n\} \subset \N$ and $m \in \N$; then $J + m = \{j_1+m,\cdots,j_n+m\}$.
Similarly, consider a set of sets of integers $\idxset = \{J_1,J_2,\cdots,J_n\}\subset \pow{\N}$.
We denote $\idxset + m$ to mean $\{J_1 + m, J_2 + m, \cdots, J_n + m\}$.

An $n\times m$ matrix of ones is $\ones_{n\times m}$.
Similarly, a matrix of zeros is $\zeros_{n\times m}$.
An $n\times n$ identity matrix is $\eye_{n}$.
Let $v \in \R^n, w \in \R^m$; we denote vector concatenation by $(v,w) \in \R^{n+m}$.
The $\diag{\cdot}$ operator places its arguments (block) diagonally on a matrix of zeros.
The $\eig{\cdot}$ operator returns a column vector containing the eigenvalues of its input matrix.
The $\det(A)$ operator returns the determinant of a square matrix $A$.
\new{For a positive semi-definite (PSD) square matrix $A \succ 0$, $B = \sqrt{A}$ is a PSD square matrix such that $B\trans B = A$.}

Let $v \in \R^n$ and $J \subset \N_n$.
Then $v\idx{J} \in \R^{\card{J}}$ is the vector of elements of $v$ indexed by $J$.
Similarly, if $A \in \R^{n\times m}$, $J_1 \subset \N_n$, and $J_2 \subset \N_m$, then $A\idx{J_1,J_2}$ is the $\card{J_1}\times\card{J_2}$ sub-matrix of $A$.
We denote $A\idx{J,:}$ as the $\card{J}\times m$ submatrix of $A$ (that is, the $J$ rows and all the columns), and $A\idx{:,J}$ similarly selects all rows and $J$ columns.

We denote ``big O'' complexity with $\bigO(\cdot)$.


\subsection{Set Representations}

An \defemph{ellipsoid} is the set
\begin{align}\label{eq:ellipsoid}
    \ellipsoidfn{c,Q} = \left\{x \in \R^\ndim\ |\ (x-c)\trans Q(x-c) \leq 1 \right\}.
\end{align}
We call $c$ its \defemph{center} and positive definite $Q \succ 0$ its \defemph{shape matrix}.
Note, some definitions use $Q\inv$ instead \cite{halder2018parameterized_ellipsoid_approx,gassmann2020scalable_ellipsoid_zono_conversion}.

A \defemph{zonotope} $\zonofn{c,G} \subset \R^n$ is a convex, symmetrical polytope parameterized by a \defemph{center} $c \in \R^n$ and a \defemph{generator matrix} $G \in \R^{n\times \ngen}$, given by
\begin{align}\label{eq:zono_defn}
    \zonofn{c,G} = \left\{c + G\beta\ |\ \norm{\beta}_\infty \leq 1\right\}.
\end{align}
That is, a zonotope is a set of convex combinations of $c$ with the columns of the matrix $G$, which we call \defemph{generators}.
We call $\beta$ the generator \defemph{coefficients}.

A \defemph{constrained zonotope} is a similar representation, but can represent any convex polytope \cite{scott2016constrained_zonotopes}.
Let $A \in \R^{k\times m}$ and $b \in \R^k$, where $k \in \N$ is the number of linear constraints.
We denote a constrained zonotope as
\begin{align}
    \conzonofn{c,G,A,b} = \left\{c + G\beta  \in \R^n\ |\ \norm{\beta}_\infty \leq 1,\ A\beta = b\right\},
\end{align}
where $c$ and $G$ are the same as for zonotopes above.

\section{Ellipsotopes}\label{sec:etopes}

In this section, we define ellipsotopes, then discuss several useful properties.
We then discuss the specific case of ellipsotopes defined using a 2-norm and conclude the section by relating ellipsotopes to other set representations.

\subsection{Definition}

To define ellipsotopes, we first introduce \defemph{index sets}.

\begin{defn}
Let $\ngen \in \N$. 
Let $\idxset \subset \pow{\N_\ngen}$ be a partition of $\N_\ngen$.
We call $\idxset$ an \defemph{index set}.
That is, $\idxset$ is a set of multi-indices such that $\N_\ngen = \union_{J \in \idxset} J$ and $J_1 \cap J_2 = \emptyset$ for any $J_1, J_2 \in \idxset$.
\end{defn}
\noindent In other words, every integer from $1$ to $\ngen$ occurs in exactly one subset $J \in \idxset$.
For example, if $\ngen = 3$, $\idxset = \{\{1,2\},\{3\}\}$ obeys the definition.
\newnc{Numerically, we store $\idxset$ as a list of lists.}


We now define ellipsotopes:

\begin{defn}\label{def:ellipsotope}
Let $c \in \R^n$, $G \in \R^{n\times\ngen}$, $A \in \R^{\ncon\times\ngen}$, $b \in \R^\ncon$, and let $\idxset$ be a valid index set.
An \defemph{ellipsotope} is a set
\begin{align}\begin{split}\label{eq:elli_defn}
    \ellifn{p}{c,G,A,b,\idxset} = \big\{c+ G\coef\ |\ &\norm{\coef\idx{J}}_p \leq 1\ \forall\ J \in \idxset\\
        &\regtext{and}\ A\coef = b \big\} \subset \R^n.
\end{split}\end{align}
A \defemph{basic} ellipsotope, $\ellifn{p}{c,G}$, has no constraints or index set.
A \defemph{constrained} ellipsotope, $\ellifn{p}{c,G,A,b}$, has no index set.
An \defemph{indexed} ellipsotope, $\ellifn{p}{c,G,\idxset}$, has no constraints.
\end{defn}
\noindent One can go further and subject different indices of $\coef$ to different $p$-norms, but we have not yet needed this in practice.

An indexed ellipsotope can be seen as an affine map of
\begin{align}\label{eq:ball_product}
    \ballprodfn{\idxset} = \left\{\coef \in \R^\ngen\ |\ \norm{\coef\idx{J}}_p \leq 1~\forall~J \in \idxset \right\},
\end{align}
which we call a \defemph{ball product} because it is the Cartesian product of $|\idxset| \in \N$ $p$-norm balls in the dimensions indexed by each $J \in \idxset$.
\new{Note, a Cartesian product of unit balls is in general not a unit ball, which necessitates using index sets such that the $p$-norm is applied to an ellipsotope's coefficients correctly.}

\new{Ellipsotopes subsume zonotopes and ellipsoids as follows:
\begin{lem}\label{lem:zonotopes_and_ellipsoids_are_ellipsotopes}
Consider the ellipsotope $E = \ellifn{p}{c,G,A,b,\idxset} \subset \R^n$.
If $\idxset = \{\{1\},\{2\},\cdots,\{\ngen\}\}$ then $E$ is a zonotope.
If $p = 2$ and $\idxset = \{\{1,2,\cdots,\ngen\}\}$ (i.e., $\card{\idxset} = 1$) then $E$ is an ellipsoid.
\end{lem}
\begin{proof}
The zonotope case follows from comparing \eqref{eq:elli_defn} to \eqref{eq:zono_defn}.
The ellipsoid case is proven later in Lem. \ref{lem:ellipsotopes_are_ellipsoids}.
\end{proof}}


\subsection{Operations on Ellipsotopes}\label{subsec:etope_ops}

Affine maps, Minkowski sums, intersections, and convex emptiness checking are the key operations that make constrained zonotopes and similar set representations useful for tasks such as reachability analysis and fault detection.
We now show that ellipsotopes are closed under these operations. 
We then provide a convex program to check whether or not an ellipsotope is empty or contains a point.
These operations are useful for collision checking an ellipsotope reachable set or detecting faults, as we show in Sec. \ref{sec:numerical_examples}.

\subsubsection{Affine Map}
The affine map of ellipsotopes follows from the definition (c.f., \cite[Prop. 1]{scott2016constrained_zonotopes}).
\new{Let $E = \ellifn{p}{c,G,A,b,\idxset}$ with $c \in \R^\ndim$ and $G \in \R^{\ndim\times\ngen}$.}
Consider an affine map parameterized by a matrix $T \in \R^{\ndim\times \ndim}$ and a translation vector $t \in \R^\ndim$.
Then
\begin{align}
    TE + t = \ellifn{p}{Tc + t,TG,A,b,\idxset}.
\end{align}

\subsubsection{Minkowski Sum}
For the ellipsotope Minkowski sum, we use index sets to apply the $p$-norm separately to the coefficients from each ellipsotope, and matrix concatenation to preserve the linear constraints from the input ellipsotopes:

\begin{prop}[Minkowski Sum]\label{prop:minkowski_sum}
Consider the ellipsotopes $E_1 = \ellifn{p}{c_1,G_1,A_1,b_1,\idxset_1}$ and $E_2 = \ellifn{p}{c_2,G_2,A_2,b_2,\idxset_2}$, both in $\R^\ndim$, with $m_1$ and $m_2 \in \N$ generators respectively.
Then the Minkowski sum $E_1 \oplus E_2 = E_\oplus$ is given by
\begin{subequations}\label{eq:mink_sum_exact}
\begin{align}
    E_\oplus &= \ellifn{p}{c_1 + c_2, [G_1, G_2], A_\oplus, b_\oplus, \idxset_\oplus},\ \regtext{with} \\
    A_\oplus &=  \diag{A_1,A_2},\quad 
    b_\oplus = \begin{bmatrix} b_1 \\ b_2 \end{bmatrix},\ \regtext{and}\\
    \idxset_\oplus &= \idxset_1\cup(\idxset_2 + \ngen_1),
\end{align}
\end{subequations}
where $\ngen_1$ is the number of generators of $E_1$ and $(\idxset_2 + \ngen_1)$ is as per Sec. \ref{subsec:notation}.
\newnc{This operation has complexity $\bigO(\ndim+\ngen_2)$.}
\end{prop}

\begin{proof}
By applying the definitions of Minkowski sums and ellipsotopes, we have
\begin{subequations}\label{eq:mink_sum_unroll}
\begin{align}
    E_1 \oplus E_2 &= \{x_1 + x_2\ |\ x_1 \in E_1,\ x_2 \in E_2 \} \\
                \begin{split}
                   &= \{c_1 + G_1 \coef_1 + c_2 + G_2 \coef_2\ |\ \norm{\coef_1\idx{J}}_p \leq 1\  \\ 
                   & \quad\quad\quad \forall\ J \in \idxset_1,\ A_1\coef_1=b_1,\ \norm{\coef_2\idx{J}}_p \leq 1,\ \\ 
                   & \quad\quad\quad \forall\ J \in \idxset_2,\ \regtext{and}\ A_2\coef_2=b_2 \}.
               \end{split}
\end{align}
\end{subequations}
Then, the proof is complete by taking $\coef = (\coef_1,\coef_2)$, expanding \eqref{eq:mink_sum_exact} using Def. \ref{def:ellipsotope}, and comparing to \eqref{eq:mink_sum_unroll}.
Notice that $\idxset$ ensures the $p$-norm constraint is applied to each subset of the coefficients of $E_\oplus$ corresponding to $E_1$ and $E_2$.
\newnc{Also notice that $c_1 + c_2$ is $\bigO(\ndim)$ and $\idxset_2 + m_1$ is $\bigO(m_2)$; all other operations are $\bigO(1)$ memory allocations.}
\end{proof}

\noindent The Minkowski sum is illustrated in Fig. \ref{fig:minkowski_sum_and_intersection} (shown in beige).

\begin{figure}[ht]
    \centering
    \includegraphics[width=0.95\columnwidth]{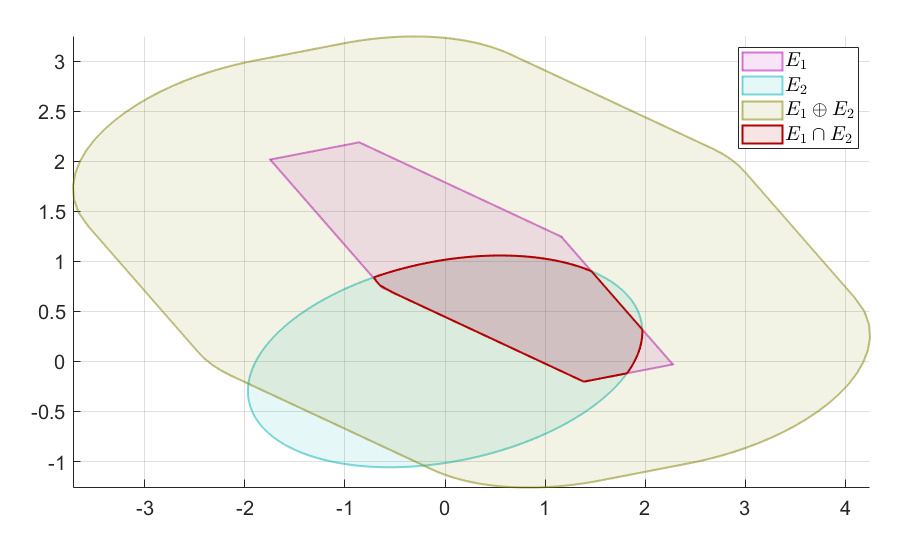}
    \caption{Example of Minkowski sum using Prop. \ref{prop:minkowski_sum} (shown in tan) and intersection using Prop. \ref{prop:intersection} (shown in red) of two ellipsotopes (shown in purple and blue).}
    \label{fig:minkowski_sum_and_intersection}
    \vspace*{-0.3cm}
\end{figure}


\subsubsection{Cartesian Product}

The Cartesian product is useful for tasks such as combining reachable sets in lower dimensions to create a single, higher-dimensional reachable set for a complex system \cite{kousik2019_quad_RTD,holmes2020reachable}.
To define this operation for ellipsotopes, let $E_1 = \ellifn{p}{c_1,G_1,A_1,b_1,\idxset_1} \subset \R^{n_1}$ and $E_2 = \ellifn{p}{c_2,G_2,A_2,b_2,\idxset_2} \subset \R^{n_2}$.
By applying similar logic to the proof of Prop. \ref{prop:minkowski_sum}, one can check that
\begin{subequations}
\begin{align}
    E_1 \times E_2 &= \ellifn{p}{c_\times,G_\times,A_\times,b_\times,\idxset_\times} \subset \R^{(n_1 + n_2)},\ \regtext{with}\\
    c_\times &= \begin{bmatrix}
        c_1 \\ c_2
    \end{bmatrix},\quad
    G_\times = \diag{G_1,G_2},\\
    A_\times &= \diag{A_1,A_2},\quad b_\times = \begin{bmatrix}
        b_1 \\ b_2
    \end{bmatrix},\ \regtext{and}\\
    \idxset_\times &= \idxset_1\cup(\idxset_2 + \ngen_1),
\end{align}
\end{subequations}
where $\ngen_1$ is the number of generators of $E_1$.
\newnc{The complexity is $\bigO(1)$ as it only consists of memory allocations.}

\subsubsection{Intersections}

By extending the constrained zonotope intersection property in \cite{scott2016constrained_zonotopes}, we define the intersection of ellipsotopes.
Note, the intersection may be empty, which one can check using Prop. \ref{prop:empty_and_point_check} below.

\begin{prop}[Ellipsotope-Ellipsotope Intersection]\label{prop:intersection}
Let $E_1 = \ellifn{p}{c_1,G_1,A_1,b_1,\idxset_1} \subset \R^\ndim$ with $\ngen_1$ generators and $A_1\in\R^{\ncon_1\times \ngen_1}$, and let $E_2 = \ellifn{p}{c_2,G_2,A_2,b_2,\idxset_2} \subset \R^\ndim$ with $\ngen_2$ generators and $A_2\in\R^{\ncon_2\times \ngen_2}$.
Then the intersection $E_1 \cap E_2$ is an ellipsotope $E_\cap$ given by
\begin{subequations}
\begin{align}
    E_\cap &= \ellifn{p}{c_1,[G_1,\zeros_{\ndim\times\ngen_2}],A_\cap,b_\cap,\idxset_\cap},\label{eq:etope_intersect_defn}\\
    A_\cap &= \begin{bmatrix}
                A_1 & \zeros_{\ncon_1\times\ngen_2} \\
                \zeros_{\ncon_2\times\ngen_1} & A_2 \\
                G_1 & -G_2
              \end{bmatrix},\ 
    b_\cap = \begin{bmatrix}
                b_1 \\
                b_2 \\
                c_2 - c_1
              \end{bmatrix},\quad\regtext{and} \\
    \idxset_\cap &= \idxset_1\cup(\idxset_2+m_1).
\end{align}
\end{subequations}
\newnc{This operation has complexity $\bigO(\ndim+m_2)$.}
\end{prop}
\begin{proof}
This follows from \cite[Prop. 1]{scott2016constrained_zonotopes} by noticing (similar to the proof of Prop. \ref{prop:minkowski_sum}) that $\idxset$ ensures that the $p$-norm constraints are applied separately to the coefficients of $E_\cap$ depending on whether they came from $E_1$ or from $E_2$.
\newnc{Notice that $c_2 - c_1$ is $\bigO(\ndim)$ and $\idxset_2 + m_1$ is $\bigO(\ngen_2)$, hence the overall complexity $\bigO(\ndim + \ngen_2)$.}
\end{proof}

\noindent This property is illustrated in Fig. \ref{fig:minkowski_sum_and_intersection} (shown in red).
Note that, since $E_1 \cap E_2 = E_2 \cap E_1$, one can choose which center to keep in \eqref{eq:etope_intersect_defn} to minimize the number of zero generators (that is, one can either add $\ngen_1$ or $\ngen_2$ generators).

Often, for reachability of hybrid systems, one must detect when a reachable set intersects a hyperplane or halfspace \cite{girard2005reachability_zono,raghuraman2020set_ops_conzono}.
We now present both of these cases for ellipsotopes.
Note, we do not assume either intersection is nonempty.

\newnc{
\begin{prop}[Ellipsotope-Hyperplane Intersection]\label{prop:hyperplane_intersection}
Consider $E = \ellifn{p}{c,G,A,b,\idxset} \subset \R^n$ with $c \in \R^n$, $G \in \R^{n\times\ngen}$, $A \in \R^{\ncon\times\ngen}$, $b \in \R^\ncon$, and $\idxset$ a valid index set.
Let $P = \hpfn{H,f} \subset \R^n$ be an affine hyperplane with $H \in \R^{k\times n}$.
Then $E \cap P = E_P$ where
\begin{subequations}
\begin{align}
    E_P &= \ellifn{p}{c, G, A_P, b_P, \idxset},\ \regtext{with}\\
    A_P &= \begin{bmatrix}
        A \\ HG
    \end{bmatrix}\quad\regtext{and}\quad
    b_P = \begin{bmatrix}
        b \\ f - Hc
    \end{bmatrix}.
\end{align}
\end{subequations}
This has complexity $\bigO(kn\ngen)$.
\end{prop}
\begin{proof}
Recall that $P = \left\{x \in \R^n\ |\ Hx = f \right\}$.
Then, if $x \in E \cap P$, there exists $\coef \in \R^\ngen$ feasible for $E$, for which
\begin{align}
    H(c + G\coef) = f \implies HG\coef = f - Hc,
\end{align}
which is the last block row of the linear constraint in $E_P$.
The complexity follows from the product $HG$.
\end{proof}
}

For the halfspace case, we adapt \cite[Theorem 1]{raghuraman2020set_ops_conzono}.

\begin{prop}[Ellipsotope-Halfspace Intersection]\label{prop:halfspace_intersection}
\new{Consider $E = \ellifn{p}{c,G,A,b,\idxset} \subset \R^n$ with $c \in \R^n$, $G \in \R^{n\times\ngen}$, $A \in \R^{\ncon\times\ngen}$, $b \in \R^\ncon$, and $\idxset$ a valid index set.
Let $S = \hspacefn{h,s} \subset \R^n$, where $h \in \R^n$ and $s \in \R$.
Then $E_S = E \cap S$ is given by
\begin{subequations}\label{eq:elli_halfplane_intersection}
\begin{align}
    E_S &= \ellifn{p}{c,[G,\ \zeros_{n\times 1}],A_S,b_S,\idxset_S}, \\
    A_S &= \begin{bmatrix}
        A & \zeros_{\ncon\times 1} \\
        h\trans G &d
    \end{bmatrix},\ 
    b_S = \begin{bmatrix}
        b \\ s - h\trans c - d
    \end{bmatrix},\\
    d &= \tfrac{1}{2}\left(s - h\trans c + \abs{h\trans G}\ones_{\ngen\times 1}\right),\ \regtext{and}\\
    \idxset_S &= \idxset\cup \{\{\ngen + 1\}\},
\end{align}
\end{subequations}
where $\abs{\cdot} \in \R^{n\times\ngen}$ is the element-wise absolute value and $\{\{\cdot\}\}$ is a singleton index set.
This has complexity $\bigO(n^2(n+m))$.}
\end{prop}
\begin{proof}
We prove this property constructively.
Recall that $S = \left\{x \in \R^n\ |\ h\trans x \leq s \right\}$.
Following the logic of Prop. \ref{prop:hyperplane_intersection}, our strategy is to add a linear constraint to the coefficients $\coef$ of $E$ constraining the resulting set to lie within the halfspace; that is, we want $h\trans(c + G\coef) \leq s$.
However, we need a slack variable to enforce this as an equality constraint:
    $h\trans(c + G\coef) + \gamma = s$,
with $\gamma \geq 0$.
We cannot add $\gamma$ directly as a coefficient to the ellipsotope, because it is unconstrained; instead, we want to bound $\gamma$ to lie within an interval, which we can map to the interval $[-1,1]$ containing a (scalar) ellipsotope coefficient.
To do this, we first find an upper bound for $\gamma$ using the fact that $E$ is compact and lies fully within a zonotope, $E \subset \zonofn{c,G}$ (see Lem. \ref{lem:generator_pop_zono_overapprox} below).
From \cite[Sec. 5.1]{girard2005reachability_zono}, we have $\gamma \leq s - h\trans c + h\trans\abs{G}\ones_{\ngen\times 1}$.
Now, we want to pick $d$ such that, for any $\gamma$, $d(\slackvar+1) = \gamma$ and $\norm{\slackvar}_p \leq 1$, where $\slackvar$ is our additional coefficient.
That is, we seek an affine transformation of the interval $[-1,1]$ to $[0,s - h\trans c + h\trans\abs{G}\ones_{\ngen\times 1})]$.
Applying interval arithmetic, we can solve $d([-1,1] + 1) = [0,\ s - h\trans c + h\trans\abs{G}\ones_{\ngen\times 1}]$ for 
\begin{align}
    d = \tfrac{1}{2}(s - h\trans c + h\trans\abs{G}\ones_{\ngen\times 1})
\end{align}
We can then construct the necessary linear equality constraint on $\coef$ and $\slackvar$ as $h\trans(c + G\coef) + d(\slackvar + 1) = s$, so
\begin{align}
    h\trans G\coef + d\slackvar = s - h\trans c - d\label{eq:halfplane_eq_cons}
\end{align}
Notice that $E_S$, as in \eqref{eq:elli_halfplane_intersection}, is the ellipsotope $E$ with one additional coefficient and the additional linear constraint in \eqref{eq:halfplane_eq_cons}, with $\idxset_S$ ensuring that $\slackvar \in [-1,1]$.
\newnc{Finally, notice that $|h\trans G|\ones_{\ngen\times 1}$ dominates the operation complexity.}
\end{proof}

\noindent To build intuition for when $E \cap S$ is empty, consider the zonotope $Z = \zonofn{c,G}$.
Notice that, if $Z \cap S = \emptyset$, then $E \cap S = \emptyset$.
In the case of the zonotope, we can interpret this to mean that the affine subspace $\hpfn{A_S,b_S} \subset \R^{(\ngen + 1)}$ does not intersect the $\infty$-norm unit ball in $\R^{(\ngen+1)}$.
Similarly for ellipsotopes, $\hpfn{A_S,b_S}$ does not intersect the ball product $\ballprodfn{\idxset_S} \subset \R^{(\ngen+1)}$.

\subsection{Emptiness and Point Containment}

Given a system's state, it is often useful to check if it lies within a specific region of state space.
Similarly, given a reachable set in state space, one may need to check if this set intersects with, e.g., an unsafe set.
Assuming ellipsotope representation of the states and sets in question, we perform the desired checks as follows, by leveraging Prop. \ref{prop:intersection} wherein the intersection of ellipsotopes is again an ellipsotope.

\begin{prop}[Emptiness and Point Containment]\label{prop:empty_and_point_check}
Consider an ellipsotope $E = \ellifn{p}{c,G,A,b,\idxset} \subset \R^n$ with $\ngen$ generators.
Assume $\hpfn{A,b} \neq \emptyset$.
Let $x \in \R^n$, and let
\begin{align}
    \costfunc{\coef} = \max_{J \in \idxset} \norm{\coef\idx{J}}_p,
\end{align}
where $\coef$ is the ellipsotope coefficient vector.
Then
\begin{align}
    E \neq \emptyset &\iff \min_{\coef \in \R^\ngen}\left\{\costfunc{\coef}\ |\ A\coef = b \right\} \leq 1\ \regtext{and} \label{prog:empty_check}\\
    x \in E &\iff \min_{\coef\in\R^\ngen}\left\{\costfunc{\coef}\ |\ 
        \begin{bmatrix} A \\ G \end{bmatrix}\coef =
        \begin{bmatrix} b \\ x - c \end{bmatrix}
    \right\} \leq 1,\label{prog:point_containment}
\end{align}
which are both convex programs.
\end{prop}
\begin{proof}
We prove the claim for \eqref{prog:empty_check}, as the claim for \eqref{prog:point_containment} then follows from Prop. \ref{prop:intersection} by checking the emptiness of $E\cap\ellifn{p}{x,\emptyarr}$.
Notice that, if $\coef \in \R^\ngen$ is feasible for the ellipsotope definition constraints in \eqref{def:ellipsotope}, then $\costfunc{\coef} \leq 1$ by construction.
Therefore, \eqref{prog:empty_check} evaluates whether or not the set $\hpfn{A,b}$ intersects $\ballprodfn{\idxset}$ (i.e., the set of feasible $\coef$ as in \eqref{eq:ball_product}).
The constraint set is nonempty by assumption and convex by inspection.
Since $\norm{\cdot}_p$ is convex, and the $\max$ of convex functions is also convex, $\costfunc{\cdot}$ is convex.
\end{proof}

\noindent \newnc{Note, the complexity of solving \eqref{prog:empty_check} or \eqref{prog:point_containment} depends on $p$ and one's choice of numerical solver.}
We find in practice that, when an ellipsotope is nonempty with $p = 2$, it takes on the order of $10^{-4}$ s to solve \eqref{prog:empty_check} (see Section \ref{sec:numerical_examples} for implementation details), but it takes two to four orders of magnitude longer for empty ellipsotopes.
However, by instead searching for a feasible $\coef$ to the constraints $A\coef = b$ and $\beta \in \ballprodfn{\idxset}$, we achieve much lower solve times in practice.
We write the search for a feasible $\coef$ as follows:

\begin{cor}[to Prop. \ref{prop:empty_and_point_check}]\label{cor:emptiness_check}
Let $E = \ellifn{p}{c,G,A,b,\idxset} \subset \R^n$ with $\ngen$ generators.
Assume $\hpfn{A,b} \neq \emptyset$.
Then
\begin{align}\label{prog:empty_check_feas_prob}
    E \neq \emptyset \iff \min_{\coef \in \R^\ngen}\left\{ \norm{A\coef - b}_2^2\ |\ \coef \in \ballprodfn{\idxset} \right\} = 0.
\end{align}
\end{cor}
\begin{proof}
This formulation follows directly from the fact that, for any feasible $\coef$, we have $A\coef = b$ and $\coef \in \ballprodfn{\idxset}$.
\end{proof}

\noindent Notice that, in the case of a constrained zonotope, \eqref{prog:empty_check_feas_prob} becomes a bounded-value least squares problem.
\new{Also note, one can reformulate point containment in \eqref{prog:point_containment} as per \eqref{prog:empty_check_feas_prob} by treating it as the intersection of an ellipsotope with a point.}
\subsection{Properties of 2-Ellipsotopes}

\new{We now discuss the special case of 2-ellipsotopes (i.e., ellipsotopes with $p = 2$), which can represent zonotopes and constrained zonotopes along with ellipsoids.}
First, we confirm that basic 2-ellipsotopes are ellipsoids and vice-versa.
Second, we notice that constrained 2-ellipsotopes are in fact basic 2-ellipsotopes.
Later, in Sec. \ref{sec:order_reduction}, we leverage these properties to create an order reduction strategy for 2-ellipsotopes.
\newnc{Note, we use matrix square roots and inverses, both of which have worst-case complexity $\bigO(n^3)$ for an $n\times n$ matrix, as they rely on Schur \cite{bjorck1983schur} or LU factorization \cite{bunch1974triangular}.}

\begin{lem}[Ellipsoid-Ellipsotope Equivalence]\label{lem:ellipsotopes_are_ellipsoids}
(Claim 1) Let $E = \ellipsoidfn{c,Q} \subset \R^n$ be an ellipsoid as in \eqref{eq:ellipsoid}.
Then $E = \ellifn{2}{c,(\sqrt{Q})\inv} \subset \R^n$.
(Claim 2) Suppose $E = \ellifn{2}{c,G} \subset \R^n$.
Then there exists $Q \in \R^{n\times n}$, $Q \succ 0$, such that $E = \ellipsoidfn{c,Q}$.
\end{lem}
\begin{proof}
\textit{(Claim 1)}
Note $(\sqrt{Q})\inv \succ 0$ exists because $Q \succ 0$.
Suppose $x \in E$, so $(x-c)\trans Q(x-c) \leq 1$.
We want to find $G$ and $\beta$ such that $(x-c) = G\coef$ and $\norm{\coef}_2 \leq 1$.
If we set $(G\coef)\trans Q(G\coef) = \coef\trans\coef$, then $G\coef = (\sqrt{Q})\inv\coef$.

\textit{(Claim 2)}
Suppose that $x \in E$, so there exists $\coef$ such that $G\coef = x - c$.
It follows from Proposition \ref{prop:empty_and_point_check} that $\coef = G\pinv(x-c)$, where $G\pinv$ is the Moore-Penrose pseudoinverse of $G$.
Since $\norm{\coef}_2^2 = \coef\trans\coef$, we have $\norm{\coef}_2^2 = (x-c)\trans(G\pinv)\trans(G\pinv)(x-c)$.
\new{Then, pick $Q = (G\pinv)\trans(G\pinv)$ (notice $Q \succeq 0$ by construction).}
\end{proof}

\noindent While these claims are well-known in the literature (e.g., \cite[(3)]{scott2016constrained_zonotopes}), we write the proof to clarify Lem. \ref{lem:basic_2_etope_order_reduc} in Section \ref{sec:order_reduction}.

Next, we find a further equivalence between constrained and basic 2-ellipsotopes.
To prove this, first, we confirm that the (nonempty) intersection of an $n$-dimensional ellipsoid with an affine subspace is a lower-dimensional ellipsoid:

\begin{lem}\label{lem:intersect_ellipsoid_with_plane}
\new{Let $B = \ballfn{2}{m}$ (the $m$-dimensional 2-norm ball) and $H = \hpfn{A,b}$ (an affine hyperplane), with $A \in \R^{n\times m}$ full row rank, $n < m$, and $b \in \R^n$.
Suppose $H$ intersects the interior of $B$ (i.e., $\card{B\cap H} > 1$).
Then $B \cap H$ is the affine image of an $(m-n)$-dimensional 2-norm ball.
That is, there exist a translation $t \in \R^m$ and a linear map $T: \R^{m-n} \to \R^m$ such that $T\ballfn{2}{m-n} + t = B \cap H \subset \R^m$.}
\end{lem}
\begin{proof}
\newnc{We prove the claim by construction.
First, let $t = A\pinv b \in H$.
Notice that $t \in H$ because $n < m$ and $A$ is full row rank; i.e., $t$ is the orthogonal projection of $0$ onto $H$.
Then, since (i) $B$ is centered at $0$, (ii) any point $p \in \bdry{B \cap H}$ has $\norm{p}_2 = 1$ by definition of $B$, and (iii) $H$ intersects the interior of $B$, it follows that there exists $c > 0$ such that $c = \norm{q - t}_2$ for any $q \in \bdry{B \cap H}$.
It also follows that $\norm{t}_2 < 1$, so $t \in B$.
In other words, $t$ is the center of an $(m-n)$-dimensional 2-norm ball defined by $B \cap H$ and embedded in $\R^m$.
To construct T, let $\{e_1,\cdots,e_{m-n}\} \subset \R^m$ be an orthonormal basis for $\ker{A}$.
Then $T$ is given by the matrix $[ce_1, \cdots, ce_{m-n}]$ (i.e., $T$ rotates $\ballfn{2}{m-n}$ to be parallel to $H$ and scales it by $c$).}
\end{proof}

\noindent \new{Intuitively, the projection of a high-dimensional ellipsoid to a lower-dimensional space is again an ellipsoid.}

\begin{lem}[Basic and Constrained 2-Ellipsotope Equivalence]\label{lem:basic_and_cons_etope_equiv}
Let $E = \ellifn{2}{c,G,A,b}$ be a nonempty constrained ellipsotope with $A \in \R^{\ncon\times \ngen}$, $b \in \R^\ncon$, and $\ncon < \ngen$.
Then there exist $c', G'$ such that $E = \ellifn{2}{c',G'}$.
\end{lem}
\begin{proof}
This follows from Lem. \ref{lem:intersect_ellipsoid_with_plane}.
Since $E$ is nonempty, we can construct an affine map parameterized by $T$ and $t$ such that $T\ballfn{2}{\ngen-\ncon} + t = B \subset \R^\ngen$.
Then, for any $\coef \in \ballfn{2}{\ngen-\ncon}$, we have $c + G(T\coef + t) \in E$.
Choose $c' = c + Gt$ and $G' = GT$ to complete the proof.
\end{proof}

Note that 2-ellipsotopes let us represent ellipsoidal Gaussian confidence level sets.
We demonstrate this via a robot path verification example in Sec. \ref{subsec:num_ex_robot_path_planning}.
\subsection{Relationships to Other Set Representations}\label{subsec:etope_relate_to_others}

Per Lem. \ref{lem:ellipsotopes_are_ellipsoids}, ellipsotopes generalize ellipsoids and, as a corollary, superellipsoids.
We see from the Definition \ref{def:ellipsotope}, specifically \eqref{eq:elli_defn} that ellipsotopes generalize (constrained) zonotopes, by comparison to \eqref{eq:zono_defn}.
And, from Lem. \ref{lem:zonotopes_and_ellipsoids_are_ellipsotopes}, if the index set is $\idxset = \{\{1\},\{2\},\cdots,\{\ngen\}\}$ for an ellipsotope with $\ngen$ generators, then the ellipsotope is also a (constrained) zonotope.

Another useful set representation is the \defemph{capsule}, often used to represent robot manipulator links for efficient collision detection \cite{macagon2003efficient_capsules,liu2016algorithmic_capsules}.
A capsule is the Minkowski sum of a line segment with a sphere, which we can represent as an ellipsotope per Lem. \ref{lem:ellipsotopes_are_ellipsoids} and Proposition \ref{prop:minkowski_sum}.
Importantly, ellipsotopes allow us to generalize capsules to Minkowski sums of line segments with, e.g., confidence level set ellipsoids of a Gaussian distribution.

Finally, one can show that all ellipsotopes are constrained polynomial zonotopes (CPZs) \cite{kochdumper2020cons_poly_zono} by extending the proof that all ellipsoids are CPZs (see the appendix).

\section{Order Reduction}\label{sec:order_reduction}

A commonly-used operation in zonotope reachability analysis is order reduction, or the approximation of a zonotope by a new zonotope with fewer generators.
This operation is necessary because reachability analysis often uses Minkowski sums, which increase the number of generators of a zonotope (or ellipsotope, per Prop. \ref{prop:minkowski_sum}).

A variety of order reduction techniques exist for zonotopes, most commonly achieved by enclosing a subset of a zonotope's generators in a bounding box, the sides of which become new generators \cite{girard2005reachability_zono,combastel2005zono_state_observer}.
This strategy can be improved or guided by a variety of heuristics \cite[Ch. 2]{althoff2010_dissertation}.
\newnc{See \cite{kopetzki2017methods,yang2018comparison} for a thorough review and comparison of methods.}
In the case of polynomial zonotopes, which are not necessarily convex, one can apply a similar strategy of overapproximating a subset of generators with a zonotope or interval \cite{kochdumper2020sparse_poly_zono,holmes2020reachable}.
For constrained zonotopes, the linear constraints necessitate alternative strategies \cite{scott2016constrained_zonotopes,raghuraman2020set_ops_conzono}.
To proceed, we discuss 2-ellipsotopes in particular, then comment on general strategies.

\subsection{Order Reduction for 2-Ellipsotopes}
\newnc{For reducing 2-ellipsotopes, we can leverage properties of ellipsoids.
Importantly, we can bound the number of generators required to exactly represent any 2-ellipsotope (Prop. \ref{prop:lift_then_reduce}).}

\subsubsection{Basic 2-Ellipsotopes}
First, we note that a basic 2-ellisotope in $\R^\ndim$ never requires more than $\ndim$ generators:

\begin{lem}[Exact Order Reduction of Basic 2-Ellipsotopes]\label{lem:basic_2_etope_order_reduc}
Let $E = \ellifn{2}{c,G} \subset \R^n$ with $G \in \R^{\ndim\times\ngen}$ full row rank and $\ngen > \ndim$.
Then $E = \ellifn{2}{c,\tilde{G}}$, where
\begin{align}\label{eq:G_basic_2_etope_order_reduc}
    \tilde{G} = \left(\sqrt{(G\pinv)\trans(G\pinv)}\right)\inv,
\end{align}
and $G\pinv$ is the Moore-Penrose pseudoinverse of $G$.
\end{lem}
\begin{proof}
This follows from Lem. \ref{lem:ellipsotopes_are_ellipsoids} by converting $E$ to an ellipsoid then back to an ellipsotope.
Note, the matrix in the outermost parentheses of \eqref{eq:G_basic_2_etope_order_reduc} is invertible because it is the square root of a positive definite matrix.
\end{proof}

\noindent Notice that $\tilde{G} \in \R^{\ndim\times\ndim}$, so $E$ needs only $\ndim$ generators.
\newnc{The complexity of \eqref{eq:G_basic_2_etope_order_reduc} is $\bigO(\ngen^3 + \ndim^3)$ as it is dominated by either the matrix product, square root, or inverse \cite{bjorck1983schur,bunch1974triangular}.}

\subsubsection{General Strategy for 2-Ellipsotopes}
Our general strategy is to treat 2-ellipsotopes as a Minkowski sum of ellipsoids.
This is because order reduction is usually necessary after several Minkowski sum operations result in a large number of generators during, e.g., reachability analysis.

To explain our approach, we consider a simple case.
Consider $E = \ellifn{2}{c,G,A,b,\idxset} \subset \R^\ndim$ with $\ngen > \ndim$ generators and with $\ncon$ linear constraints.
Suppose that we can write
    $E = E_1 \oplus E_2$
where $E_1 = \ellifn{2}{c_1,G_1,A_1,b_1}$ with $\ngen_1$ generators and $E_2 = \ellifn{2}{c_2,G_2,A_2,b_2}$ with $\ngen_2$ generators.
Notice that $\ngen = \ngen_1 + \ngen_2$.
Our goal is to find $\tilde{E}$ for which
    $\tilde{E} = \ellifn{2}{c',G'} \supset E$.

First, by Lem. \ref{lem:basic_and_cons_etope_equiv}, we can find $t_1$ and $T_1$ such that 
    $E_1 = \ellifn{2}{c_1 + G_1t_1,\ G_1T_1}$,
and similarly for $E_2$.
Then, per Lem. \ref{lem:ellipsotopes_are_ellipsoids}, we can find $Q_1$ to represent $E_1$ as an ellipsoid, 
    $E_1 = \ellipsoidfn{c_1 + G_1t_1,\ Q_1}$,
and similarly we can find $Q_2$ for $E_2$.

We now apply the method in \cite{halder2018parameterized_ellipsoid_approx} to create a minimum-volume outer ellipsoid (MVOE) $E\reduce \supseteq E_1 \oplus E_2$. 
That is, we can write
    $E\reduce = \ellipsoidfn{c\reduce,Q\reduce} \supset E_1 \oplus E_2$,
\newnc{By Lem. \ref{lem:ellipsotopes_are_ellipsoids}, we have $E\reduce \supset E$.
By Lem. \ref{lem:basic_2_etope_order_reduc}, $E\reduce$ needs no more than $\ndim < \ngen$ generators.
Therefore, we can choose $\tilde{E} = E\reduce$.}

\new{
We note that a variety of techniques exist to tightly approximate an MVOE \cite{halder2018parameterized_ellipsoid_approx,kurzhanski2000ellipsoidal,kurzhanskiy2006ellipsoidal_toolbox}.
In this work, we apply \cite{halder2018parameterized_ellipsoid_approx}, which is equivalent to the parameterization in \cite{kurzhanskiy2006ellipsoidal_toolbox}, but enables fixed-point iteration to find the MVOE more quickly than standard semi-definite programming approaches.
}

\subsubsection{Choosing Which Ellipsoids to Overapproximate}\label{subsubsec:choosing_MVOE_ellipsoids}

The above example considered an ellipsotope created as the Minkowski sum of a pair of ellipsoids, so the order reduction strategy was to overapproximate this sum with a single ellipsoid.
We now extend this idea to the case when an ellipsotope is a Minkowski sum of many ellipsoids.

First, we set up our assumptions.
Consider again the ellipsotope $E = \ellifn{2}{c,G,A,b,\idxset} \subset \R^\ndim$ with $\ngen$ generators.
Assume that we can write $E$ as the Minkowski sum of several basic 2-ellipsotopes, which we call \defemph{component ellipsoids}:
\begin{align}
    E = E_1 \oplus E_2 \oplus \cdots \oplus E_\nsum,\label{eq:etope_mink_sum_component_etopes}
\end{align}
for some $\nsum \in \N$.
That is, each $E_i = \ellifn{2}{c_i,G_i}$.
Notice that $E$ requires at most $\nsum\times\ndim$ generators.

Now, suppose that we want to find $\tilde{E}$ such that $\tilde{E} \supset E$ and $\tilde{E}$ has $\ngen - \ndim$ generators; in other words, we want to reduce the number of 2-ellipsotopes in \eqref{eq:etope_mink_sum_component_etopes} by one.
To do so, we choose $i, j \in \N_\nsum$ and construct $E\reduce = E_i \oplus E_j$ such that
\begin{align}
    \tilde{E} = \bigg(\bigoplus_{l \in \N_\nsum\setminus\{i,j\}} E_l \bigg) \oplus E\reduce.
\end{align}

The question is then how to choose $i$ and $j$.
Our goal for choosing $i$ and $j$ is to minimize the conservativeness introduced by overapproximating $E_i \oplus E_j$.
The most straightforward option is to choose the $(i,j)$ pair for which the MVOE has the smallest volume.
\new{For an ellipsoid $E = \ellipsoidfn{c,Q} \subset \R^\ndim$, the volume is proportional to $\det(Q\inv)$ \cite[Sec. I]{halder2018parameterized_ellipsoid_approx}.
So, by Lem. \ref{lem:ellipsotopes_are_ellipsoids} and because all component ellipsoids are in $\R^\ndim$, we can choose those for which $\det(((G\pinv)\trans(G\pinv))\inv)$ is smallest.
Note, the determinant takes $\bigO(n^3)$ time via LU decomposition \cite{bunch1974triangular}.}

\new{
However, it may be computationally expensive to compute the MVOE for every possible pair (of which there are $r^2$ for $r$ component ellipsoids).
Instead, we apply a heuristic.
Let $Q_1$ and $Q_2$ be ellipsoid shape matrices as in \eqref{eq:ellipsoid}.
We use \cite{halder2018parameterized_ellipsoid_approx} to solve for a value $\zeta$ such that the MVOE's shape matrix is
\begin{align}
    Q_\oplus = \left((1+\zeta)Q_1\inv + (1+\tfrac{1}{\zeta})Q_2\inv\right)\inv.
\end{align}
We find empirically that $\zeta \approx 1$ in most cases.
So, we use the following heuristic to pick $(i,j)$:
\begin{align}\label{eq:order_reduc_heur_2_etope}
    (i,j) = \argmax_{i,j \in \N_\nsum}\  \left(\det\!\left((2Q_i\inv + 2Q_j\inv)\inv\right)\right)\inv
\end{align}
We evaluate the quality of this heuristic in Sec. \ref{subsec:order_reduction_examples}; in short, it correlates strongly with the volume of the MVOE.}


\subsubsection{Identifying Component Ellipsoids}

In Section \ref{subsec:etope_ops}, we found that intersections between ellipsotopes, hyperplanes, and halfspaces all introduce linear constraints.
Strategies exist to conservatively simplify these linear constraints for constrained zonotopes \cite{scott2016constrained_zonotopes,raghuraman2020set_ops_conzono}. 
For 2-ellipsotopes, we can instead use the index set and constraints to identify component ellipsoids. 

Notice that all intersections introduce a new block row to the ellipsotope constraints (see Props. \ref{prop:intersection}, \ref{prop:hyperplane_intersection}, and \ref{prop:halfspace_intersection}), while placing any existing constraints either block-diagonally (in the case of ellipsotope-ellipsotope intersection) or with zero-padding (for halfspace intersection).
Furthermore, the ellipsotope's index set contains the indices of the columns corresponding to the constraints that existed before the intersection procedure.
Therefore, given an arbitrary ellipsotope, if we identify indices in the index set that correspond to a block-diagonal arrangement of linear constraints, then we can extract the component ellipsoids and reduce them with Lem. \ref{lem:basic_and_cons_etope_equiv}.

To illustrate this idea with an example, consider an ellipsotope $E = \ellifn{2}{c,G,A,b,\idxset}$ with $\ngen$ generators.
Suppose that $A = \diag{A_1,A_2} \in \R^{2\times\ngen}$, $A_1 \in \R^{1\times\ngen_1}$, and $A_2 \in \R^{1\times\ngen_2}$.
Also suppose $\idxset = \{\N_{\ngen_1},\N_{\ngen_2}+\ngen_1\}$.
Then
\begin{align}\begin{split}
    E =~&\ellifn{2}{c,G\idx{:,\N_{\ngen_1}},A_1,b\idx{\N_{\ngen_1}}} \oplus\\
        &\oplus\ellifn{2}{c,G\idx{:,\N_{\ngen_2}},A_2,b\idx{\N_{\ngen_2}}}.
\end{split}\end{align}
In other words, we have broken $E$ into two component ellipsoids, which we can then reduce as above.

\subsubsection{Lift-then-Reduce}
\new{It may not be possible to identify component ellipsoids if the constraint matrix does not have a block-diagonal structure.
However, we can apply a lifting strategy \cite[Prop. 3]{scott2016constrained_zonotopes} to shift constraints into the generator matrix, producing an ellipsotope in the form of \eqref{eq:etope_mink_sum_component_etopes}, albeit in higher dimensions.
By leveraging the properties of 2-ellipsotopes, we then have the following bound:

\begin{prop}[Lift-then-Reduce]\label{prop:lift_then_reduce}
Let $E = \ellifn{2}{c,G,A,b,\idxset} \subset \R^\ndim$ with $\ngen$ generators and $\ncon$ constraints.
Then there exists an ellipsotope $E\reduce = \ellifn{2}{c,G\reduce,A\reduce,b,\idxset\reduce} \subset \R^\ndim$ such that $E = E\reduce$ and $E\reduce$ has no more than $(\ndim+\ncon)\cdot\card{\idxset}$ generators.
\end{prop}
\begin{proof}
First, notice that $x \in E$ if and only if
\begin{align}\label{eq:lifted_ellipsotope}
    \begin{bmatrix} x \\ \zeros \end{bmatrix} \in 
        E_+ := \ellifn{2}{\begin{bmatrix} c \\ -b \end{bmatrix},
            \begin{bmatrix} G \\ A \end{bmatrix},\ \idxset}.
\end{align}
This is because there exists $\coef \in \ballprodfn{\idxset}$ with $A\coef = b$ such that $x = c + G\coef$, so {\footnotesize $\begin{bmatrix} x \\ b \end{bmatrix} = \begin{bmatrix} c \\ \zeros \end{bmatrix} + \begin{bmatrix} G \\ A \end{bmatrix}\coef$}.
Denote $E_+ = \ellifn{2}{c_+,G_+,\idxset}$, and let $J \in \idxset$.
From Prop. \ref{prop:minkowski_sum}, we have that $G_+\idx{:,J}$ corresponds to a component ellipsoid as in \eqref{eq:etope_mink_sum_component_etopes}.
If $\card{J} \geq \ndim + \ncon$, we can reduce $G_+\idx{:,J}$ exactly according to Lem. \ref{lem:basic_2_etope_order_reduc}; that is, every component ellipsoid of $E_+$ never needs more than $\ndim+\ncon$ generators.
To complete the proof, pick $G\reduce$ (resp. $A\reduce$) as the first $\ndim$ rows (resp. last $\ncon$ rows) of every reduced $G_+\idx{:,J}$ (which are concatenated horizontally after reduction) for each $J \in \idxset$, and set $J\reduce = \{\delta,\cdots,\delta+\ndim+\ncon-1\}$ with $\delta \in \N$ chosen appropriately for each reduced $G_+\idx{:,J}$.
Construct $\idxset\reduce$ from all such $J\reduce$.
\end{proof}
\noindent We call $E_+$ the \defemph{lifted} ellipsotope.
Note, one can further reduce $E_+$ by applying the MVOE strategy from Sec. \eqref{subsubsec:choosing_MVOE_ellipsoids} before constructing $E\reduce$ (in which case we overapproximate $E$).
The complexity of Prop. \ref{prop:lift_then_reduce} is $\bigO(\card{\idxset}\cdot(\ndim+\ncon)^3)$ since it requires $\card{\idxset}$ matrix products, inverses, and square roots in the worst case.
}

\subsection{General Strategies for Order Reduction}


We now briefly discuss order reduction when $p \neq 2$.
In short, strategies from the literature for zonotopes and constrained zonotopes still apply to ellipsotopes.
We leave strategies that leverage the $p$-norm structure to future work.

\subsubsection{Leveraging Component Zonotopes}

We noted above that order reduction for an arbitrary 2-ellipsotope follows from treating it as a Minkowski sum of component ellipsoids.
For a basic $p$-ellipsotope, we can adopt a similar strategy by considering component zonotopes.

First, notice that, by making a single generator's $p$-norm constraint independent from all other generators, we overapproximate an ellipsotope.
We call this \defemph{popping} a generator:
\begin{lem}[Generator Popping]\label{lem:generator_pop_zono_overapprox}
Consider the indexed ellipsotope $E = \ellifn{p}{c,G,\idxset}$.
Consider an arbitrary $J \in \idxset$ and suppose $j \in J$.
Define $\tilde{J} = (J\setminus\{j\})$ and $\tilde{\idxset} = (\idxset\setminus J) \cup \tilde{J} \cup \{j\}$.
Then $E \subset \tilde{E}$ where $\tilde{E} = \ellifn{p}{c,G,\tilde{\idxset}}$.
\end{lem}
\begin{proof}
For any feasible $\beta$, $\norm{\beta\idx{J}}_p \leq \norm{\beta\idx{\tilde{J}}}_p + |\beta\idx{j}|$ by the triangle inequality.
\end{proof}
\noindent In Lem. \ref{lem:generator_pop_zono_overapprox}, we have \defemph{popped} the $j$\ts{th} generator.

Then, a strategy for order reduction is as follows.
Suppose $E \subset \R^\ndim$ has $\ngen$ generators, and we seek to remove $n\reduce$ of them.
First, we pop the $n\reduce + \ndim$ smallest (in the 2-norm) generators.
Let $G = [G\keep, G\reduce]$ where $G\reduce$ contains these $n\reduce + \ndim$ generators; note we can reorder $G$ in this way without loss of generality.
Let $Z\reduce = \ellifn{p}{\zeros,G\reduce,\{\{1\},\{2\},\cdots,\{n\reduce\}\}}$, which is a zonotope by Lem. \ref{lem:zonotopes_and_ellipsoids_are_ellipsotopes}.
If we pop the $G\reduce$ generators, then $E = \ellifn{p}{c,G\keep,\idxset\keep} \oplus Z\reduce$, where $\idxset\keep$ is the original index set with the indices corresponding to $G\reduce$ removed, and then reorganized to match $G\keep$.
Finally, we can apply zonotope order reduction \cite{combastel2005zono_state_observer,girard2005reachability_zono,althoff2010_dissertation} to find an $\ndim$-dimensional interval $\tilde{Z}\reduce \supseteq Z\reduce$, which can be represented as an ellipsotope with $\ndim$ generators per Lem. \ref{lem:zonotopes_and_ellipsoids_are_ellipsotopes}.
Note that generator popping enables further simplification for $p = 2$, as one can overapproximate the zonotope created by all popped generators by a single ellipsoid using the technique in \cite{gassmann2020scalable_ellipsoid_zono_conversion}.
This MVOE approximation can be made tighter by applying \cite[Lem. 3]{malik2006gap_for_MVOE_of_zonotope} if $\ndim \approx n\reduce$.

\subsubsection{Constraint Reduction}

For $p \neq 2$, the result in Lem. \ref{lem:basic_and_cons_etope_equiv} no longer holds; that is, the intersection of a superellipsoid with an affine subspace is not always an affine map of a lower-dimensional superellipsoid, which can be seen by considering the $\infty$-norm ball intersecting a plane.
However, the constraint reduction strategies from \cite{scott2016constrained_zonotopes} and \cite{raghuraman2020set_ops_conzono} still apply.
For example, we can eliminate a constraint by adapting \cite[Prop. 5]{scott2016constrained_zonotopes}:

\begin{prop}\label{prop:constraint_dualization}
Let $E = \ellifn{p}{c,G,A,b,\idxset} \subset \R^\ndim$ with $\ngen$ generators and $\ncon$ constraints.
Let $\Gamma \in \R^{\ndim\times\ncon}$ and $\Lambda \in \R^{\ncon\times\ncon}$.
Then
\begin{align}
    E \subseteq \tilde{E} = \ellifn{p}{c + \Gamma b, G - \Gamma A, A - \Lambda A, b - \Lambda b,\idxset}.
\end{align}
\end{prop}
\begin{proof}
Let $x \in E$, so $\exists\ \coef \in \R^\ngen$ such that $x = c + G\coef$ and $A\coef = b$.
Then $x = c + G\coef + \Gamma(b - A\coef)$ and $A\coef = b + \Lambda(b - A\coef)$.
\end{proof}

\noindent By choosing $\Lambda$ as a matrix of zeros with a single one on the diagonal, one can zero out a row of $[A,b]$ to eliminate a constraint.
One can also use Prop. \ref{prop:constraint_dualization} to remove a constraint and a generator by choosing $\Gamma$ and $\Lambda$ as in \cite[Sec. 4.2]{scott2016constrained_zonotopes}.
\section{Numerical Examples}\label{sec:numerical_examples}

We now demonstrate properties and uses of ellipsotopes: we illustrate fault detection, assess the speed of the emptiness check, verify collision-avoidance for robot path planning under uncertainty, and assess our order reduction heuristic.
We use MATLAB 2020b to implement all examples\footnote{All code used for figures and examples is available online at \newline \url{https://github.com/Stanford-NavLab/ellipsotopes}.}.

\subsection{Fault Detection}
\new{We implement the set-based fault detection example based on \cite[Section 6]{scott2016constrained_zonotopes}, for which a 2-D, linear nominal model is given but a faulty model (i.e., with slightly different system matrices) is propagated.
A set-based estimator is propagated using the faulty model, and a point containment check is performed at each timestep on samples drawn from the true model.
The goal is to detect the fault (i.e., discrepancy between the nominal and faulty model) in the fewest timesteps.
This is run on a 6-core, 3.4 GHz desktop with 32 GB RAM.

The discrete-time system dynamics follow the form
\begin{align}\begin{split}
    x(t) &= A_i x(t-1) + B_i u(t-1) + D_i w(t-1)\\
    y(t) &= C_i x(t) + v(t)
\end{split}\end{align}
with state $x(t) \in \R^2$, control input $u(t-1)\in\R$, measurement $y(t) \in \R^2$, disturbance $w(t-1) \in \R^2$, and measurement error $v(t)\in\R^2$, and $i=1,2$ distinguishes the nominal and faulty models respectively.
The nominal and faulty system matrices are given by
\begin{subequations}
\begin{align}
    A_1 &= \begin{bmatrix} \Delta_t & 0 \\ 0 & \Delta_t \end{bmatrix} \
    B_1 = \begin{bmatrix} \Delta_t \\ 0 \end{bmatrix} \
    D_1 = \begin{bmatrix} -0.1 & -0.2 \\ -0.2 & 0.1 \end{bmatrix}  \\
    A_2 &= \begin{bmatrix} 2\Delta_t & 0 \\ 0 & 2\Delta_t \end{bmatrix} \
    B_2 = \begin{bmatrix} 2\Delta_t \\ 0 \end{bmatrix} \
    D_2 = \begin{bmatrix} -0.2 & -0.2 \\ -0.1 & 0.1 \end{bmatrix} 
\end{align}
\end{subequations}
where $\Delta_t = 0.001$ s.
Both models have $C = \eye_2$.

When ellipsotopes are used to represent constrained zonotopes, using the same order reduction strategy and random seed, we verify that both representations take an average of 45.2 timesteps (standard deviation 24.5 timsteps) to detect a fault over 5 simulation runs of 100 iterations each, as expected.
The constrained zonotope implementation, via the CORA 2021 toolbox \cite{althoff2015introduction}, takes an average of $18.2$ ms per timestep (std. $11.9$ ms), while our ellipsotope implementation using Cor. \ref{cor:emptiness_check} averages $4.0$ ms per timestep (std. $1.2$ ms).
The implementation difference is solely the index set, which takes negligible time to maintain in practice.

To illustrate the utility of ellipsotopes, we replace the original noise zonotope with an ellipsoid (i.e., a basic ellipsotope), and overapproximate this ellipsoid with zonotopes of $\ngen = 3, 8, 14$ generators, using the method in \cite{gassmann2020scalable_ellipsoid_zono_conversion} as implemented in CORA 2021.
We apply the same order reduction strategy \cite[Sec. 4.2]{scott2016constrained_zonotopes} for ellipsotopes and constrained zonotopes.
In this case, ellipsotopes detect the fault in an average of $26.6$ timesteps (std. $32.4$ timesteps) with an average runtime of $4.24$ ms per timestep (std. $1.73$ ms), whereas constrained zonotopes with $\ngen = 3$ fails to detect the fault, and detects the fault in an average of $59.0$ timesteps (28.8 ms per timestep) for both $\ngen = 8, 14$ due to overapproximation of the set estimate.
Using only ellipsoids, via CORA 2021, we fail to detect the fault, and average $0.101$ s per timestep (std. $0.0094$ s) due to computing MVOEs.}

\subsection{Emptiness Checking}\label{subsec:emptiness_check_experiment}

We now evaluate the speed of checking if an ellipsotope is empty using Cor. \ref{cor:emptiness_check}, using an 8-core, 2.4 GHz laptop with 32 GB RAM.
We apply Cor. \ref{cor:emptiness_check} because we find in practice that solving the feasibility problem \eqref{prog:empty_check_feas_prob} is orders of magnitude faster than solving \eqref{prog:empty_check} from Prop. \ref{prop:empty_and_point_check}.
This speed-up is because there is often a continuum of optimal solutions to \eqref{prog:empty_check_feas_prob}, but only one optimal solution to \eqref{prog:empty_check}.

\new{Our evaluation method is as follows.
First, we generate 10 random 2-ellipsotopes with $\ndim \in \{2, 8, 14\}$ for each $\ngen = 1,2,\cdots,20$ generators (each generator of length no more than $1/\ngen$) and $\ncon = 1$ constraint.
Then, we set $b = \zeros_{\ncon\times 1}$ or $b = 2\ngen\cdot\ones_{\ncon\times 1}$ (to ensure emptiness).
Finally, we solve \eqref{prog:empty_check_feas_prob} using the \texttt{fmincon} SQP algorithm (default tolerances) and an initial guess of $\coef_0 = A\pinv b$; we measure solve time with \texttt{timeit}.}

The results, summarized in Fig.~\ref{fig:emptiness_check_stats}, show that it takes on the order of $10^{-4}$ s to confirm that an ellipsotope is nonempty, but $10^{-2}$ s to identify that an ellipsotope is empty.
This is because the initial guess of $A\pinv b$ is often a feasible solution to \eqref{prog:empty_check_feas_prob}, so the solver can terminate on the first iteration.
\new{Note, the number of generators is the size of the decision variable of \eqref{prog:empty_check_feas_prob}, so we see similar solve time for varying $\ndim$.}

\begin{figure}[ht]
    \centering
    \includegraphics[width=\columnwidth]{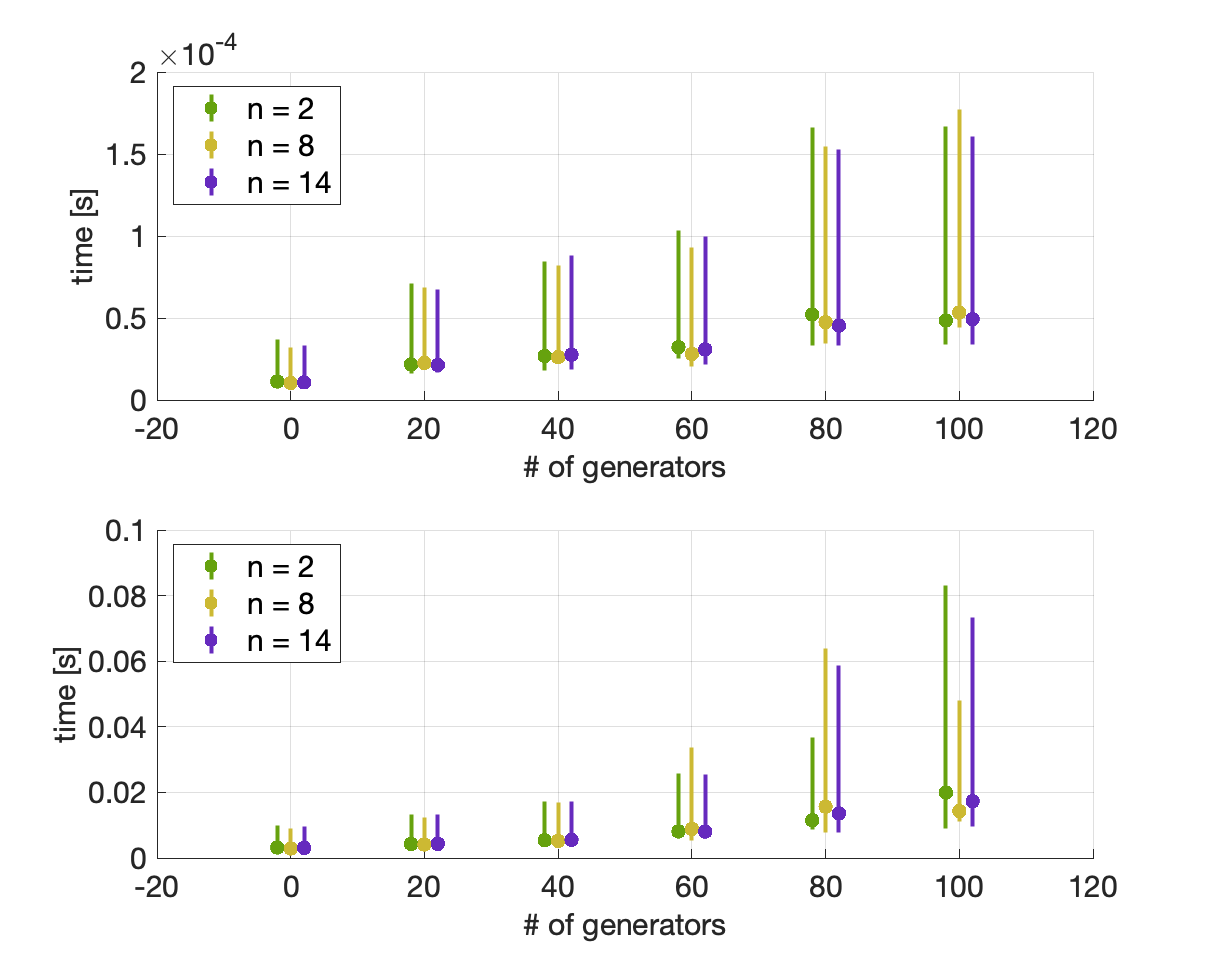}
    \caption{Timing results for solving the ellipsotope emptiness check \eqref{prog:empty_check_feas_prob} as in Section \ref{subsec:emptiness_check_experiment}.
    The top (resp. bottom) subplot shows the emptiness check times for nonempty (resp. empty) ellipsotopes.
    The dots show the mean solve time, and the bars show min/max times.
    Empty ellipsotopes take longer because they require multiple iterations to solve \eqref{prog:empty_check_feas_prob} instead of terminating upon finding a feasible solution.}
    \label{fig:emptiness_check_stats}
    \vspace*{-0.3cm}
\end{figure}

\subsection{Robot Path Verification}\label{subsec:num_ex_robot_path_planning}

We now present a path verification example in which ellipsotopes are used to represent the reachable set of the combined volume of a robot's body and state uncertainty.
This illustrates the practicality of the ellipsotope Minkowski sum, intersection, and emptiness check.
\new{To demonstrate that ellipsotopes can provide tighter reachable sets than zonotopes or ellipsoids, we also compute the reachable sets for the same trajectory using both zonotopes and ellipsoids via CORA 2021 \cite{althoff2015introduction}.}
We use a 6-core, 3.4 GHz desktop with 32 GB RAM.

\subsubsection{System Dynamics and Measurements}

We consider a robot with a box-shaped rigid body with width $w\rob$ and length $l\rob$, and represent it with an indexed 2-ellipsotope:
\begin{align}\begin{split}\label{eq:robot_body}
    E\rob = \ellifn{2}{\zeros_{2\times 1},\ \tfrac{1}{2}\diag{w\rob,l\rob},\ \{\{1\},\{2\}\}},
\end{split}\end{align}
We model the system with discrete-time, nonlinear dynamics and measurements.
In particular we consider a Dubins car model with state $x(t) = [x_1(t), x_2(t), \theta(t)]\trans$, input $u(t) = [v(t), \omega(t)]\trans$ and center-of-mass equations of motion
\begin{subequations}
\begin{align}
    x_1(t) &= x_1(t-1) + v(t-1)\cos(\theta(t-1))\Delta_t + w_1(t), \\
    x_2(t) &= x_2(t-1) + v(t-1)\sin(\theta(t-1))\Delta_t + w_2(t), \\
    \theta(t) &= \theta(t-1) + \omega(t-1)\Delta_t + w_3(t),
\end{align}
\end{subequations}
where $p(t) = [x_1(t), x_2(t)]\trans$ is the robot's center-of-mass position and $\theta(t)$ is its heading at time $t \in \N$.
The process noise is $w(t) \sim \pN(\zeros,Q)$ where $Q \in \R^{3\times 3}$ and $Q \succ 0$.
The control inputs are longitudinal speed $v(t)$ and yaw rate $\omega(t)$.
Time is discretized by $\Delta_t = 0.1$ s.

The robot's measurements consist of 4 ranges to beacons placed at fixed, known locations, as well as a heading measurement, all with additive Gaussian noise.
Range measurements that are taken when $x_1(t) < 30$ have noise variance of \SI{0.4}{\meter}, while measurements taken when $x_1(t) \ge 30$ (shown shaded in light red in Fig.~\ref{fig:reach_comparison}) have a higher variance of \SI{10.0}{\meter}.

\subsubsection{Reachability under Position Uncertainty}

The robot tracks a nominal trajectory $\check{x}$ with a linear state estimator and controller, as in \cite{bry2011rapidly} and \cite{shetty2020_stoch_reach}. 
At time $t$ the state estimator provides an uncertain robot state as a Gaussian distribution $\pN(\mu(t),\Sigma(t))$.
We assume the position and heading covariance are decoupled, such that we can decompose $\mu(t)$ and $\Sigma(t)$ into position and heading components $\mu(t) = [\mu_p(t), \mu_\theta(t)]\trans$ and $\Sigma(t) = \diag{\Sigma_p(t), \Sigma_\theta(t)}$.
Now consider the $\alpha$-probability confidence level set of the robot's uncertain position, $E\uncrt$, for which $P(p(t)\in E\uncrt) \ge \alpha$.
Letting $\epsilon = -2 \log(1-\alpha)$, we represent $E\uncrt$ as an ellipse, $E\uncrt = \{x + \check{p}(t)\ |\ x\trans(\epsilon\Sigma_p(t))^{-1}x \le 1\}$.
Then, with Lem. \ref{lem:ellipsotopes_are_ellipsoids}, we represent this ellipse as a 2-ellipsotope $E\uncrt(t) = \ellifn{2}{\check{p}(t),(\epsilon\Sigma_p(t))^{1/2}}$.
Given some initial state estimation covariance $\Sigma_0$, we propagate state uncertainty along the nominal trajectory according to \cite[Equations (17)-(21) and (33)]{bry2011rapidly}, and obtain the associated $\alpha$-confidence ellipses that enclose the center-of-mass trajectory of the robot, under uncertainty due to noisy dynamics and measurements. 

\subsubsection{Handling Robot Body and Heading Uncertainty}

To account for the robot's body, we cannot simply Minkowski sum the $E\rob$ ellipsotope with the $E\uncrt$ ellipsotope, because we must account for heading uncertainty.
We do so by first taking the $\alpha$-confidence interval, $(\check{\theta}-\Delta_\theta, \check{\theta}+\Delta_\theta)$, of the distribution $\pN(\hat{\theta},\Sigma_\theta)$ of heading ($\theta$) estimates.
\new{Next, to overbound the area swept out by the robot's body over this range of angles, we create an ellipsotope as the intersection of the circumscribing circle of the robot's body with four halfspaces, shown in Fig. \ref{fig:heading_uncertainty} as cyan dashed lines, found analytically using $\check{\theta}\pm\Delta_\theta$.}
Then, for each timestep of the trajectory, we Minkowski sum this ellipsotope with the center-of-mass confidence ellipse from position uncertainty propagation to obtain our final reachable set.
\begin{figure}[t]
    \centering
    \includegraphics[width=0.8\columnwidth]{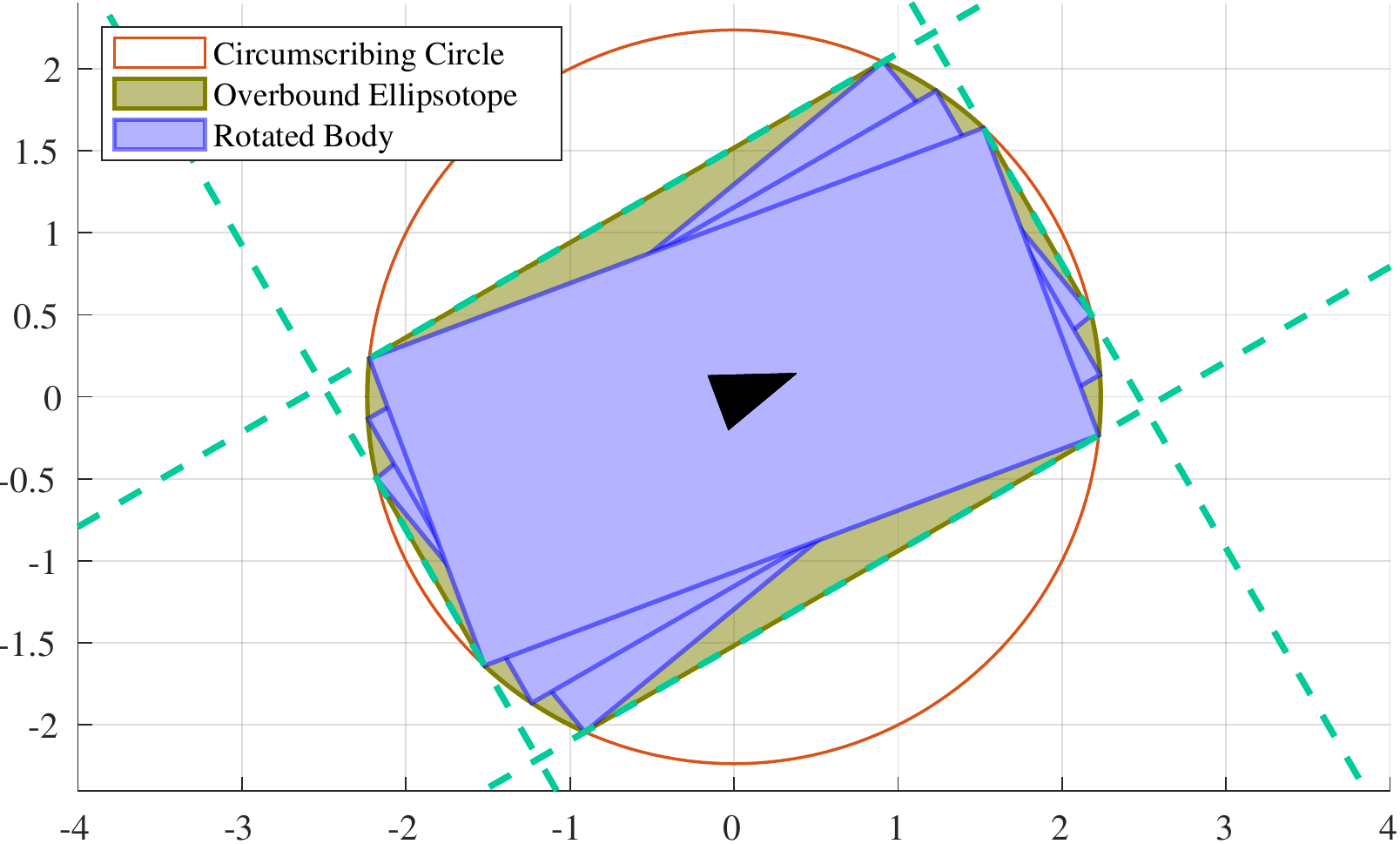}
    \caption{Construction of an ellipsotope (dark yellow) which encloses a rectangular robot body (as in Sec.~\ref{subsec:num_ex_robot_path_planning}) for an interval of headings (three possible rotations of the body shown in blue, with heading as a black arrow).}
    \label{fig:heading_uncertainty}
    \vspace*{-0.2cm}
\end{figure}

\subsubsection{Evaluation Metrics}
For each of the 127 timesteps of the nominal trajectory, we compute the intersection between the reachable set and each obstacle.
We then solve the emptiness check in Cor.~\ref{cor:emptiness_check} to assess if the reachable set is in collision.
To collision check the comparison ellipsoid and zonotope reachable sets, we use CORA \cite{althoff2015introduction}.

We compute the total area of each 2-D reachable set to assess conservativeness.
For ellipsotopes, we approximate area by sampling points from the boundary, constructing a polygon from the sampled points, then computing the area of the polygon.
For zonotopes and ellipsoids we use the CORA built-in functions for computing area.

\subsubsection{Results and Discussion}

The ellipsoid, zonotope, and ellipsotope reachable sets are shown in Fig. \ref{fig:reach_comparison}.
The ellipsotope reachable set is computed in 44.8 \si{ms} and collision checked in 1.7262 s.
We consider a 12.7 s long trajectory, so we can validate it with ellipsotopes faster than real time.
The zonotope reachable set is collision checked in 1.0832 s and the ellipsoid reachable set in 1.3676 s.
The zonotope reachable set has an area of 152.98 \si{m^2}, the ellipsoids 178.06 \si{m^2}, and the ellipsotopes 111.22 \si{m^2}.
Thus, ellipsotopes maintain comparable collision checking speed but provide a tighter reachable set.
Also note, this example is an improvement over \cite{shetty2020_stoch_reach}, since we \textit{exactly represent} the confidence bounds of the uncertain position and heading states as ellipsotopes, instead of overapproximating the bounds with zonotopes.

\begin{figure}[t]
    \centering
    \includegraphics[width=\columnwidth]{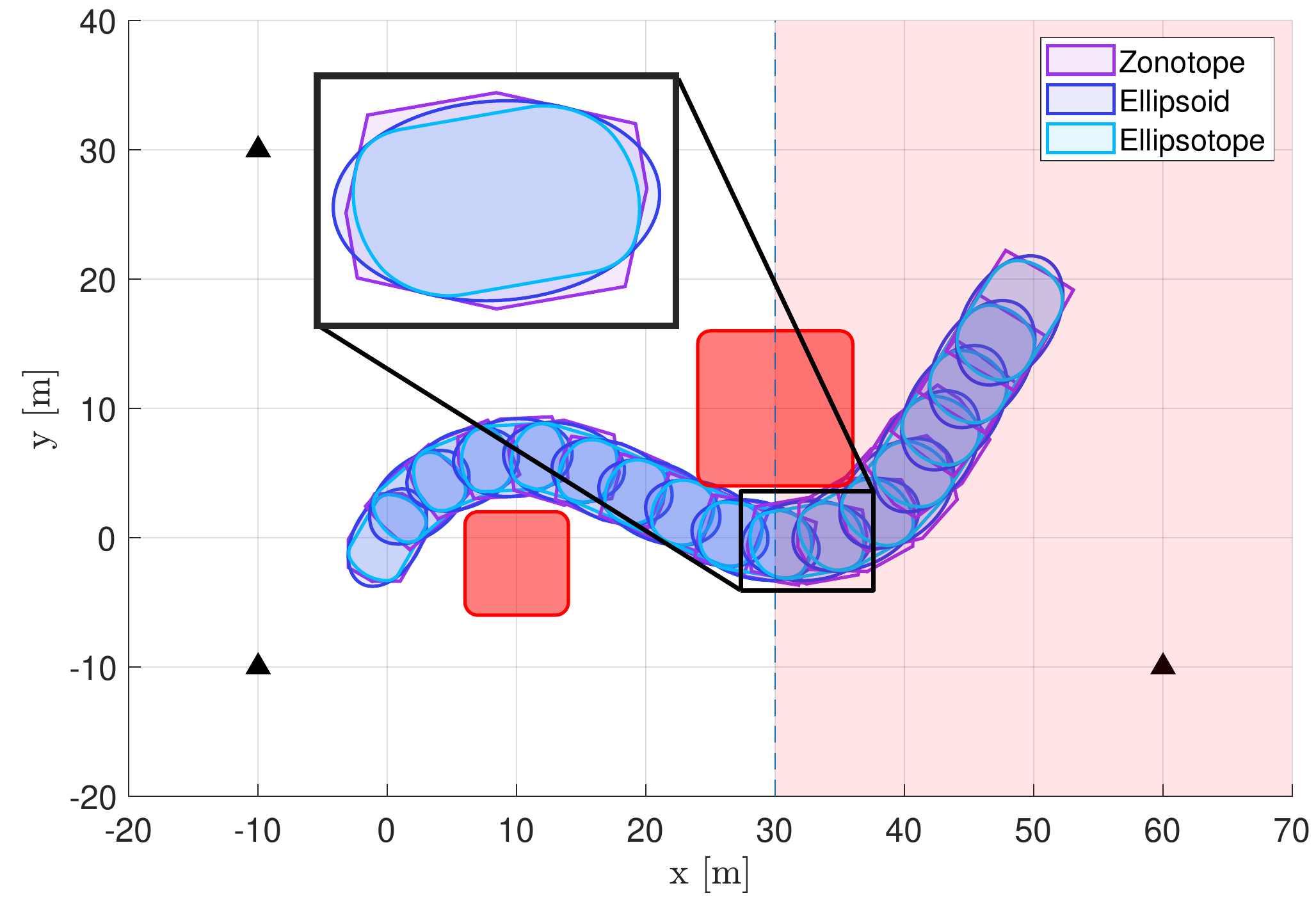}
    \caption{Comparison of reachable sets represented by zonotopes, ellipsoids, and ellipsotopes.
    The ellipsotopes more tightly bound the robot's body, as seen in the zoom window. }
    \label{fig:reach_comparison}
    \vspace*{-0.3cm}
\end{figure}

\subsection{Order Reduction Heuristic}\label{subsec:order_reduction_examples}

\new{
Finally, we assess the heuristic in Sec. \ref{subsubsec:choosing_MVOE_ellipsoids} for order reduction of a 2-ellipsotope in $\R^n$.
We use an 8-core, 2.4 GHz laptop with 32 GB RAM.
Note, our code base has examples of the other techniques from Sec. \ref{sec:order_reduction}.
For $\ndim \in \{2,8,14,50,100\}$, we create $50$ random 2-ellipsotopes as a Minkowski sum of 6 ellipsoids with random generator matrices $G \in \R^{\ndim\times\ndim}$, where each element of $G$ is drawn uniformly from $\big[\tfrac{-1}{\sqrt{\ndim}},\tfrac{1}{\sqrt{\ndim}}\big]$.
For each ellipsotope and every possible pair of component ellipsoids, we compute the true MVOE as per \cite{halder2018parameterized_ellipsoid_approx} and our heuristic value in \eqref{eq:order_reduc_heur_2_etope}.
Across all $\ndim$, despite the wide variety of generator matrices, our heuristic correlates strongly with the volume of the MVOE, but computes nearly an order of magnitude faster across all dimensions.
Note, for 50- and 100-D, $r^2 \approx 1.000$ and the mean heuristic evaluation time is on the order of $10^{-3}$ s.
Fig. \ref{fig:order_reduction_MVOE_heur_14D} shows $\ndim = 14$.
}

\begin{figure}[ht]
    \centering
    \includegraphics[width=0.8\columnwidth]{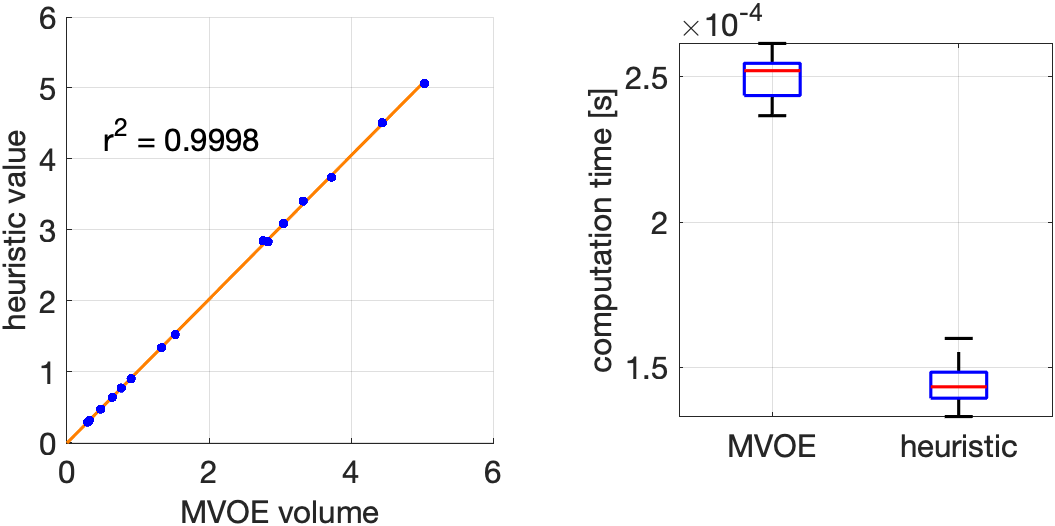}
    \caption{MVOE volume vs. heuristic \eqref{eq:order_reduc_heur_2_etope} for 14-D ellipsotopes.
    The heuristic strongly correlates with the true MVOE volume but computes faster, as per the box-and-whisker plot.}
    \label{fig:order_reduction_MVOE_heur_14D}
    \vspace*{-0.3cm}
\end{figure}
\section{Conclusion}\label{sec:conclusion}

This work introduced ellipsotopes, a novel set representation created by generalizing the $\infty$-norm that defines zonotopes and constrained zonotopes.
We showed that this set representation is closed under the operations critical to reachability analysis and fault detection: affine transformations, Minkowski sum, and intersection.
Since ellipsotopes can grow in complexity similar to zonotopes when used for reachability analysis, we discussed several order reduction strategies.
We also demonstrated the utility of ellipsotopes via numerical examples, and illustrated their importance via a literature comparison to other set representations.
For future work, we will formalize a stochastic variant of ellipsotopes and explore further applications in reachability and fault detection.


\renewcommand{\bibfont}{\normalfont\footnotesize}
{\renewcommand{\markboth}[2]{}
\printbibliography}

\appendix
\subsection{Ellipsotope Visualization}

Zonotopes can be difficult to visualize due to the exponential number of faces as a function of the number of generators.
Ellipsotopes are a further challenge because we are now concerned with plotting an affine image of, in the most general sense, the intersection of a hyperplane with a Cartesian product of high-dimensional $p$-norm balls.
We now describe several ``tricks'' to facilitate visualization.

For this appendix, consider $c \in \R^n$, $G \in \R^{n\times m}$, $A \in \R^{\ncon\times \ngen}$, and $b \in \R^\ncon$; and, let $\idxset$ be a valid index set.
For a set $S \subset \R^n$, we denote its boundary by $\bdry{S}$.

\subsubsection{High-dimensional Ball Boundaries}

The first trick that we leverage is that, since an ellipsotope is an affine image of a convex shape in high dimensions, the boundary of the image is a subset of the image of the boundary.
Consequently, our strategy for visualization is to generate points on the \emph{boundary} of the feasible generator coefficients, then map them through the generator matrix.
While many of these points may not lay on the boundary of the ellipsotope, the convex hull of these points is a visually-acceptable approximation of the ellipsotope in practice.

To proceed, we first need the following lemma.

\begin{lem}\label{lem:bdry_under_affine_image}
Suppose $C \in \R^n$ is a compact, convex set.
Let $M: \R^n \to \R^m$, with $m < n$, be a surjective linear map.
Suppose $x \in \bdry{(MC)}$.
Then there exists $y \in \bdry{C}$ such that $x = My$.
\end{lem}
\begin{proof}
The preimage $M\inv x$ is a linear subspace of $\R^n$ that intersects $C$, and therefore intersects $\bdry{C}$.
\end{proof}

Now, to approximate the boundary of a high-dimensional ball, suppose that $Y \subset \R^n$ is a finite set of sampled points with $\card{Y} = \nplot \in \N$.
To map $Y \to \bdry{\ballfn{p}{n}}$, let $\tobdryfn{p}{\cdot}: \R^n \to \R^n$ for which
\begin{align}
    \tobdryfn{p}{y} = \frac{y\idx{i}}{\norm{y\idx{i}}_p}.
\end{align}
Suppose $E = \ellifn{p}{c,G}$.
Let $X = c + G\cdot\tobdryfn{p}{Y}$.
Then, applying Lemma \ref{lem:bdry_under_affine_image}, $\convhull{X}$ approximates $\bdry{E}$.

\subsubsection{Ball Product and Affine Subspace Ray Tracing}\label{app:subsec:ball_product_affine_subsp_int}

Now we generalize the previous approach to generate points on the boundary of the intersection of the ball product $\ballprodfn{\idxset}$ and the affine subspace $\hpfn{A,b}$.
Let $E = \ellifn{p}{c,G,A,b,\idxset}$.
To plot this ellipsotope, our goal is to first pick $\nplot \in \N$ coefficients $\coef_i$, $i = 1,\cdots,\nplot$, such that
\begin{align}
    \coef_i \in \bdry{\ballprodfn{\idxset} \cap \hpfn{A,b}}.
\end{align}
In other words, these coefficients obey the constraints
\begin{subequations}\label{cons:point_on_ball_prod_lin_space_bdry}
\begin{align}
    \forall~J \in \idxset,~\norm{\coef_i\idx{J}}_p &\leq 1,\\
    \exists~K \in \idxset\ \regtext{s.t.}\ \norm{\coef_i\idx{K}}_p &= 1,\ \regtext{and}\\
    A\coef = b.
\end{align}
\end{subequations}
We can then approximate the ellipsotope as
\begin{align}
    E \approx \convhull{\{c + G\coef_i\}_{i=1}^\nplot}.
\end{align}
We generate these points by tracing rays outwards in $\R^\ngen$ from a point inside $\ballprodfn{\idxset}$ until they contact the boundary of $\ballprodfn{\idxset}\cap\hpfn{A,b}$.

To generate a single $\coef_i$, first let ${u_i} \in \ker{A}$ be a random unit vector in the nullspace of $A$.
Let ${\coef_0} \in \ballprodfn{\idxset}$, which can be found by applying Cor. \ref{cor:emptiness_check}; note that $\coef_0 = A\pinv b$ is often such a point.
Let $J \in \idxset$.
We then solve
\begin{align}
    \norm{\alpha_i{u_i}\idx{J} + {\coef_0}\idx{J}}_p = 1
\end{align}
for $\alpha_i \in \R$ and set $\coef_i = \alpha_i{u_i} + {\coef_0}$ as a point that is guaranteed to obey \eqref{cons:point_on_ball_prod_lin_space_bdry}.
Let $\varphi(\alpha_i,{u_i},{\coef_0},J) = \norm{\alpha_i{u_i}\idx{J} + {\coef_0}\idx{J}}_p^p - 1^p$.
Notice that, for any $J \in \idxset$,
\begin{align}
    \varphi(\alpha_i,{u_i},{\coef_0},J) &= \left(\sum_{j \in J}\left((\alpha_i{u_i}\idx{j})^p + ({\coef_0}\idx{j})^p\right)\right) - 1\\
    \begin{split}
        &= \binom{p}{0}\alpha_i^p\sum_{j \in J}({u_i}\idx{j})^{p}({\coef_0}\idx{j})^0 +\\
        &\quad+\binom{p}{1}\alpha_i^{p-1}\sum_{j \in J}({u_i}\idx{j})^{p-1}({\coef_0}\idx{j})^{1}~+~\cdots\\
        &\quad\cdots+\binom{p}{p}\alpha_i^0\sum_{j \in J}({u_i}\idx{j})^{0}({\coef_0}\idx{j})^p - 1,
    \end{split}
\end{align}
where we have applied the binomial theorem to expand the coefficients.
Since $\varphi(\cdot,{u_i},t,J)$ is a univariate polynomial in $\alpha_i$, we can solve for $\alpha_i$ efficiently.
Critically, since ${\coef_0} \in \ballprodfn{\idxset}$, the direction $u_i$ points ``outward'' towards the boundary, so the smallest solution $\alpha_i$ is a point on the boundary; in other words, $\coef_i$ obeys \eqref{cons:point_on_ball_prod_lin_space_bdry}.

\subsubsection{Ray Tracing}

We find in practice that plotting 2-D ellipsotopes with more than 5 generators with the above methods is computationaly expensive, taking several seconds to generate a single plot.
Furthermore, the above methods result in many unused points (that is, points on the boundary of the feasible set that are mapped to the interior of the ellipsotope, and therefore not used for plotting).
To address this, we pose a convex program to identify points on the boundary of the ellipsotope directly in its workspace.
In particular, we maximize the length of a ray extending from a point in the ellipsotope in an arbitrary direction while constraining it to lie within the ellipsotope.

We set up to perform ray tracing as follows.
Let $E = \ellifn{p}{c,G,A,B,\idxset} \subset \R^n$ be an ellipsotope with $\ngen$ generators.
Consider a ray
\begin{align}
    \rayfn{x,g} = \left\{x + \lambda g\ |\ \lambda \geq 0 \right\} \subset \R^n,
\end{align}
where $g \in \R^n$ is arbitrary and $x \in E$.
We find $x$ as any feasible point in the ellipsotope by applying the strategy above in Appendix \ref{app:subsec:ball_product_affine_subsp_int} to find a feasible coefficient $\coef$, then setting $x = c + G\coef$.
Note, we cannot always set $x = c$, because it is possible that $c \not\in E$, which occurs when $\zeros \not\in \hpfn{A,b}$.

Finally, to perform ray tracing, we solve
\begin{subequations}\label{prog:ray_tracing_etope_bdry}
\begin{align}
    \max_{\lambda \geq 0,\ \coef \in \R^\ngen}\quad&\lambda \\
    \st\quad &\norm{\coef\idx{J}}_p \leq 1~\forall~J \in \idxset,\\
        &A\coef = b,\ \regtext{and}\\
        &c + G\coef = x + \lambda g,
\end{align}
\end{subequations}
which is convex and always feasible if $E \neq \emptyset$.
By solving \eqref{prog:ray_tracing_etope_bdry} for a variety of $g$, we can sample the boundary of the ellipsotope.
In practice,  we sample $g$ uniformly from the boundary of the 2-D or 3-D unit sphere and solve \eqref{prog:ray_tracing_etope_bdry} once for each sample.
\subsection{Minimum Volume Outer Ellipsoids}\label{app:MVOE}

We use the following methods to compute minimum volume outer ellipsoids (MVOEs) for the Minkowski sum of ellipsoids and for zonotopes.

\subsubsection{MVOE of Ellipsoid Minkowski Sum}
We apply the method in \cite{halder2018parameterized_ellipsoid_approx}.
Consider the pair of ellipsoids $E_1 = \ellipsoidfn{c_1,Q_1}$ and $E_2 = \ellipsoidfn{c_2,Q_2}$ in $\R^\ndim$.
Let $\lambda = \eig{Q_1Q_2\inv} \in \R^\ndim$.
Let $\zeta_0 \in \R$ and consider the fixed-point iteration
\begin{align}\label{eq:MVOE_iteration}
    \zeta_{n+1} =  \left(\frac
        {\sum_{i=1}^\ndim \frac{1}{1+\zeta_n\lambda\idx{i}}}
        {\sum_{i=1}^\ndim \frac{\lambda\idx{i}}{1+\zeta_n\lambda\idx{i}}}
    \right)^{\frac{1}{2}}.
\end{align}
Define $\zeta$ as the limit of \eqref{eq:MVOE_iteration} as $n \to \infty$.
Then
\begin{align}
    E_1 \oplus E_2 \subseteq \ellipsoidfn{c_1 + c_2,\ Q_\oplus},
\end{align}
where
\begin{align}
    Q_\oplus = \left((1 + \tfrac{1}{\zeta})Q_1\inv + (1 + \zeta)Q_2\inv\right)\inv.
\end{align}
See \cite[Sec. III-B and Sec. IV-B]{halder2018parameterized_ellipsoid_approx} for derivation and analysis of the above algorithm, and in particular \cite[Thm. 1]{halder2018parameterized_ellipsoid_approx} for a proof of convergence.

We find $\zeta$ numerically by applying \cite[(17)]{halder2018parameterized_ellipsoid_approx}.
First, we set a tolerance $\varepsilon \approx 10^{-10}$.
Then, we iterate \eqref{eq:MVOE_iteration} starting from $\zeta_0 = 0$ until
\begin{align}
    \sum_{i=1}^\ndim \frac{1 - \zeta^2\lambda\idx{i}}{1 + \zeta\lambda\idx{i}} \leq \varepsilon.
\end{align}

\subsubsection{Overapproximating the MVOE of a Zonotope}

We apply the method in \cite[Thm. 1]{gassmann2020scalable_ellipsoid_zono_conversion}.
Let $Z = \zonofn{c,G} \subset \R^\ndim$ be a zonotope with $\ngen \in \N$ generators.
To overapproximate the MVOE, we first solve an SDP \cite[Lem. 3]{gassmann2020scalable_ellipsoid_zono_conversion}:
\begin{subequations}
\begin{align}
    r = \min_{\lambda \geq 0,\ \lambda \in \R^\ngen}\quad&\ones_{\ngen\times 1}\trans\lambda \\
    \regtext{s.t.}\quad\quad&\diag{\lambda} - G_0\trans G_0 \succeq 0,
\end{align}
\end{subequations}
where
\begin{align}
    G_0 = E_0^{-\tfrac{1}{2}}\quad\regtext{and}\quad E_0 = \ngen G G\trans.
\end{align}
Then an outer approximation of the MVOE is given by
\begin{align}
    Z \subset \ellipsoidfn{c,rE_0}.
\end{align}
The MVOE approximation can be made tighter by applying \cite[Lem. 3]{malik2006gap_for_MVOE_of_zonotope} in the case when $\ndim \approx \ngen$.

\subsection{Constrained Polynomial Zonotopes}



We can show that every ellipsotope is a constrained polynomial zonotope (CPZ) similar to showing that every ellipsoid is a CPZ \cite{kochdumper2020cons_poly_zono}.
First, we introduce polynomial notation: for a vector $v \in \R^n$ and an integer matrix $M \in \N^{m\times n}$, let $v^M \in \R^m$ denote a vector for which
\begin{align}
    v^M\idx{j} = \prod_{i = 1}^n (v\idx{i})^{M\idx{j,i}}
\end{align}
with $j = 1,\cdots, m$.
Now, given $c \in \R^n$, $G \in \R^{n\times m}$, $X \in \N^{m\times m}$, $A \in \R^{k\times m}$, $b \in \R^k$, and $D \in \N^{k \times m}$, a CPZ is the set
\begin{align}\begin{split}
    \conpolyzonofn{c,G,X,A,b,D} = \bigg\{&c + G\beta^X\ |\ 
        \norm{\beta}_\infty \leq 1\\
        &\regtext{and}\ A\beta^D - b = 0    
    \bigg\}.
\end{split}\end{align}
Now, consider the basic case of $E = \ellifn{p}{c,G} = \left\{c + G\coef\ |\ \norm{\coef}_p \leq 1 \right\}$ with $m$ generators.
Add a slack coefficient $\slackvar \in \R$ to write
\begin{align}
      \begin{split}
    E  = \Big\{c + [G, \zeros_{n\times 1}](\coef, \slackvar)\ |\ 
        &\norm{\coef}_p^p + 0.5\slackvar = 0.5,\ \regtext{and}\\
        &\norm{(\coef,\slackvar)}_\infty \leq 1 \Big\}.
      \end{split}
\end{align}
Then, it follows that
\begin{subequations}
\begin{align}
    E &= \conpolyzonofn{c,[G,\zeros],X,A,b,D},\ \regtext{with} \\
        X &= \eye_m,\ 
            A = [\ones_{1\times m}, 0.5],\\
        b &= 0.5,\ \regtext{and}\ 
            D = [p\cdot\ones_{1\times m},\ 1]\trans.
\end{align}
\end{subequations}
Adding linear constraints or an index set on the coefficients of $E$ necessitates only minor changes to $A$, $b$, and $D$ in the CPZ formulation.
\subsection{Convex Hulls}

We adapt \cite[Theorem 5]{raghuraman2020set_ops_conzono} to overapproximate the ellipsotope convex hull.
The convex hull of $A \cup B$ is $\convhull{A\cup B} = \left\{\lambda a + (1-\lambda) b\ |\ \lambda \in [0,1],\ a \in A,\ b \in B \right\}$.

\begin{prop}[Convex Hull Overapproximation]\label{prop:convhull}
Consider $E_1 = \ellifn{p}{c_1,G_1,A_1,b_1,\idxset_1}$ and $E_2 = \ellifn{p}{c_2,G_2,A_2,b_2,\idxset_1}$, with $c_1, c_2 \in \R^\ndim$, $G_1 \in \R^{\ndim\times \ngen_1}$, $G_2 \in \R^{\ndim\times \ngen_2}$, $A_1 \in \R^{\ncon_1\times \ngen_1}$, $A_2 \in \R^{\ncon_2\times \ngen_2}$, $b_1 \in \R^{\ncon_1}$, and $b_2 \in \R^{\ncon_2}$. Let $m_3 = m_1 + m_2$.
The convex hull $\convhull{E_1 \cup E_2}$ is overapproximated by the ellipsotope $E\ch = \ellifn{p}{c\ch,G\ch,A\ch,b\ch,\idxset\ch}$ with
\begin{subequations}\label{eq:convhull}
\begin{align}
    c\ch &= \tfrac{1}{2}(c_1 + c_2), \\
    G\ch &= \left[G_1, G_2, \tfrac{1}{2}(c_1 - c_2), 0\right] \in \R^{\ndim \times (3m_3 + 1)}, \\
    A\ch &= \begin{bmatrix*}[r]
        	A_1 & \zeros & -\frac{1}{2}b_1 & 0 \\
        	0 & A_2 & \frac{1}{2}b_2 & 0 \\
        	A_{3,1} & A_{3,2} & A_{3,0} & \eye 
        	\end{bmatrix*} \in \R^{(\ncon_1+\ncon_2+2m_3) \times (3m_3+1)}, \\
    b\ch &= \begin{bmatrix*}[r]
        	\frac{1}{2}b_1 \\ \frac{1}{2}b_2 \\ -\frac{1}{2}1
        	\end{bmatrix*} \in \R^{\ncon_1+\ncon_2+2m_3}, \\
    A_{3,1} &= \begin{bmatrix*}[r]
               \eye \\ -\eye \\ \zeros \\ \zeros
               \end{bmatrix*}, \quad
    A_{3,2} = \begin{bmatrix*}[r]
               \zeros \\ \zeros \\ \eye \\ -\eye 
               \end{bmatrix*}, \quad
    A_{3,0} = \begin{bmatrix*}[r]
               -\frac{1}{2}\ones \\ -\frac{1}{2}\ones \\ \frac{1}{2}\ones \\ \frac{1}{2}\ones
               \end{bmatrix*}, \\
    \idxset\ch &= \left\{\idxset_1, \idxset_2 + m_1, \{m_3+1\}, \{m_3+2\},\cdots,\{3m_3+1\} \right\},
\end{align}
\end{subequations}
(i.e. $\convhull{E_1 \cup E_2} \subseteq E\ch$).
\end{prop}
\begin{proof}
We must show that, for any $x\in\convhull{E_1 \cup E_2}$, $x\in E\ch$. 
If $x\in\convhull{E_1 \cup E_2}$, then $\exists \ x_1\in E_1$, $x_2\in E_2$, and $\lambda \in [0,1]$ such that $x = \lambda x_1 + (1 - \lambda) x_2$.
\begin{subequations}\label{eq:in_convhull_conds}
\begin{align}
    &x = \lambda x_1 + (1 - \lambda) x_2, \quad \lambda \in [0,1] \\
    &x_1 = c_1 + G_1 \gamma_1, \quad ||\gamma_1\idx{J}||_p \leq 1 \ \forall\ J \in \idxset_1,\ A_1\gamma_1 = b_1, \\
    &x_2 = c_2 + G_2 \gamma_2, \quad ||\gamma_2\idx{J}||_p \leq 1 \ \forall\ J \in \idxset_2,\ A_2\gamma_2 = b_2. 
\end{align}
\end{subequations}
To show $x\in E\ch$, we must show that there exists $\coef\in\R^{3m_3+1}$ such that $x = c\ch + G\ch\coef$ with $\norm{\coef\idx{J}}_p \leq 1\ \forall\ J \in \idxset\ch$ and $A\ch\coef = b\ch$.
Following the approach in \cite[Theorem 5]{raghuraman2020set_ops_conzono}, pick $\coef = (\coef_1, \coef_2, \coef_0, \slackvar)$ for which
\begin{subequations}\label{eq:convhull_coeffs}
\begin{align}
    &\coef_1 = \lambda\gamma_1, \quad \coef_2 = (1-\lambda)\gamma_2,\quad \coef_0 = 2\lambda - 1,\ \regtext{and} \\
    &\slackvar = -\tfrac{1}{2}\ones_{2\ngen_3\times 1} - (A_{31}\coef_1 + A_{32}\coef_2 + A_{30})\coef_0,
\end{align}
\end{subequations}
where $\coef_1\in\R^{m_1}$, $\coef_2\in\R^{m_2}$, $\coef_0\in\R$, and $\slackvar\in\R^{2m_3}$.
Substituting \eqref{eq:convhull_coeffs} into \eqref{eq:in_convhull_conds}, we can rewrite (\ref{eq:in_convhull_conds}) as
\begin{align}
    x &= \lambda (c_1 + G_1 \gamma_1) + (1 - \lambda) (c_2 + G_2 \gamma_2) \\
      &= \frac{c_1}{2}(1+\coef_0) + G_1\coef_1 + \frac{c_1}{2}(1-\coef_0)c_2 + G_2\coef_2\ \regtext{and}\\
      &= \frac{c_1+c_2}{2} + G_1\coef_1 + G_2\coef_2 + \frac{c_1-c_2}{2}\coef_0.\label{eq:convhull_final_step}
\end{align}
Next, plugging $\coef$ into \eqref{eq:convhull}, we have
\begin{subequations}\label{eq:convhull_etope_def_rollout}
\begin{align}
    &x = \frac{c_1+c_2}{2} + G_1\coef_1 + G_2\coef_2 + \frac{c_1-c_2}{2}\coef_0 + \zeros\slackvar, \\
    &\norm{\coef_1\idx{J}}_p \leq 1 \ \forall\ J \in \idxset_1, \ \norm{\coef_2\idx{J}}_p \leq 1 \ \forall\ J \in \idxset_2,  \\
    &|\coef_0| \leq 1,\ \norm{\slackvar}_\infty \leq 1,\ \regtext{and}\ A\ch\beta = b\ch.
\end{align}
\end{subequations}
Notice that $\idxset\ch$ enforces $\norm{\slackvar}_\infty \leq 1$ by construction.
By comparing \eqref{eq:convhull_final_step} to \eqref{eq:convhull_etope_def_rollout}, the proof is complete.
\end{proof}

\end{document}